\newif\ifmai\maifalse
\newif\iffull\fulltrue
\documentclass[a4paper,fleqn]{cas-sc}

\usepackage{float}

\usepackage{algorithmicx, algpseudocode}
\usepackage{algorithm}
\usepackage{amsmath,empheq}
\makeatletter
\def\maketag@@@#1{\hbox{\m@th\normalfont\normalsize#1}}
\makeatother
\usepackage{listings}
\usepackage{tikz}
\usepackage{yfonts}
\usepackage{fancyvrb}
\usepackage{subcaption}
\usepackage{tikz}
\usepackage{url}
\usetikzlibrary{arrows.meta,automata,positioning}
\usepackage[many]{tcolorbox}
\tcbset{colframe=black!70,colback=white,boxrule=0.4pt,arc=2pt,boxsep=2pt,left=2pt,right=2pt,top=1pt,bottom=1pt}
\usepackage{xcolor}
\newdefinition{definition}{Definition}
\newdefinition{remark}{Remark}
\newtheorem{proposition}{Proposition} 
\newdefinition{example}{Example}
\usepackage{caption}
\AtBeginDocument{
	\lstset{
		frame=none,
		backgroundcolor={},
		aboveskip=5pt,
		belowskip=5pt
	}
}

\usepackage{chngcntr}
\usepackage{apptools}
\newcommand{\highlight}[1]{#1}

\newtheorem{theorem}{Theorem}
\newtheorem{lemma}[theorem]{Lemma}
\newcommand{\Adj}{\mathrm{adj}}
\newcommand{\wt}{\mathrm{w}}
\newcommand{\is}{\mathrm{is}}
\newcommand{\fs}{\mathrm{fs}}
\newcommand{\ia}{\mathrm{ia}}
\newcommand{\fa}{\mathrm{fa}}

\newcommand{\N}{\mathbb{N}}
\newcommand{\Z}{\mathbb{Z}}
\newcommand{\QQ}{\mathbb{Q}}
\newcommand{\reals}{\mathbb{R}}

\newcommand{\J}{\mathcal{S}}
\newcommand{\A}{\Puiz}
\newcommand{\TT}{\mathcal{T}}
\newcommand{\C}{\mathbb{C}}
\newcommand{\pGC}{\mbox{$\mathrm{pGCL}$}}
\usetikzlibrary{automata,positioning}
\newcommand{\vsp}{\vspace*{.5cm}}
\newcommand{\Num}{T}
\newenvironment{proof_of}[1]{\noindent{{{{\sc Proof of}   {#1}.}}}~}{ \hfill $\blacktriangleleft$ \vskip .2mm}
\newcommand{\mik}[1]{}
\newcommand{\replace}[1]{#1}
\newcommand{\poc}{pOC}

\makeatletter
\newcommand{\pushright}[1]{\ifmeasuring@#1\else\omit\hfill$\displaystyle{#1}$\fi\ignorespaces}
\newcommand{\pushleft}[1]{\ifmeasuring@#1\else\omit$\displaystyle{#1}$\hfill\fi\ignorespaces}
\makeatother

\usepackage{chngcntr}
\usepackage{apptools}

\newlength{\continueindent}
\usepackage{etoolbox}
\makeatletter
\newcommand*{\ALG@customparshape}{\parshape 2 \leftmargin \linewidth \dimexpr\ALG@tlm+\continueindent\relax \dimexpr\linewidth+\leftmargin-\ALG@tlm-\continueindent\relax}
\apptocmd{\ALG@beginblock}{\ALG@customparshape}{}{\errmessage{failed to patch}}
\makeatother

\usepackage[numbers]{natbib}

\usepackage{amsmath}
\DeclareMathOperator{\lcm}{lcm}
\DeclareMathOperator{\adj}{adj}
\def\tsc#1{\csdef{#1}{\textsc{\lowercase{#1}}\xspace}}
\tsc{WGM}
\tsc{QE}

\newproof{proof}{Proof}
\usepackage{float}

\usepackage{relsize}%

\ExplSyntaxOn
\cs_gset:Npn \__first_footerline:
{
	\group_begin:
	\small
	\sffamily
	\__short_authors:
	\group_end:
}
\ExplSyntaxOff
\begin{document}
\let\WriteBookmarks\relax
\def\floatpagepagefraction{1}
\def\textpagefraction{.001}

\shorttitle{Algebraic analysis of probabilistic programs}

\shortauthors{M. Boreale, L. Collodi, A. Pompa Di Gregorio}

\title[mode = title]{On the algebraic analysis of runtime distribution  of probabilistic programs}

%

\author[1]{Michele Boreale}



\ead{michele.boreale@unifi.it}




\author[1]{Luisa Collodi}


\ead{luisa.collodi@unifi.it}




\author[1]{Alessandro {Pompa Di Gregorio}}


\ead{alessandro.pompa@edu.unifi.it }



\affiliation[1]{organization={Università  degli Studi di Firenze},
	country={Italy}}

\cortext[1]{Corresponding author}



\begin{abstract}
We present an algebraic method for analysing probabilistic programs with counters and discrete states,
\emph{Generalised Constant Probability (GCP)} programs. We define the    operational semantics of GCP in terms of
the   runs of probabilistic one-counter automata (\poc), a type of probabilistic pushdown automata.  We characterise the resulting
(sub-)probability generating function (pgf) $\Delta(z)$ as an  {algebraic} function, representable via the roots
of a \emph{kernel polynomial} associated with the program. Next, we provide algorithms that, leveraging this information, compute under mild algebraic conditions the  dominant singularities and the {exact} radius of convergence of $\Delta(z)$, leading to an exact asymptotic expansion   and to exponential  bounds for its coefficients. 
Our approach is sound for
GCP programs and complete for the single-state subclass.
\end{abstract}



\maketitle





\section{Introduction}\noindent
Much research effort has been devoted to methods for analysing termination and running time of
probabilistic programs. For almost-sure termination, a wide range of techniques, based on e.g.\ ranking
supermartingales, have been proposed --- see e.g.~\cite{AgrawalChatterjeeNovotny17,Chatterjee21,FioritiHermanns15}
and references therein. For expected running time, proof systems and algorithms have been put forward for
specific classes of programs, often with the goal of obtaining upper bounds,
see e.g.~\cite{CarbonneauxNgoHoffmann18,Giesl19,KaminskiKatoenMathejaOlmedo18}. All these techniques
are necessarily incomplete, as the termination problem is undecidable in general.

Beyond almost-sure termination and bounds on expectation, one often needs more precise analytic
information about the probability distribution of running time, such as its asymptotic behaviour and  sharp (exponential) bounds on its tails.
Identifying an interesting class of programs for
which this information can be obtained algorithmically is the goal of the present paper.
We target
programs featuring both a \emph{counter} $d$ and a \emph{state} variable $q$ that can be updated
probabilistically. An example is shown in
Listing~\ref{lst:running_example} (our running example).
In essence, we allow for different regimes of
probabilistic updates for $d$ and $q$, depending on the current value of $q$. This format generalises
the \emph{Constant Probability (CP)} programs of \cite{Giesl19}, so we refer to it as the
\emph{Generalised Constant Probability (GCP)} format.

Our analysis of GCP programs follows a distinctively algebraic approach. We define the semantics of a
program directly via the weighted runs of its associated
\highlight{pushdown automaton (PDA), specifically a probabilistic one-counter automaton (\poc) \cite{pOCACM}}. This induces a (sub-)probability
generating function (pgf) $f(z)=\sum_{n\ge0} p_n z^n$ whose coefficients $p_n$ give the probability
that the program executes exactly $n$ loop iterations. We then leverage tools from analytic combinatorics
\cite{Flajo} to analyse such pgf's. We decompose and analyse them in terms of runs of a
\highlight{\poc}, in turn equivalent to a type of random walks 
called  \emph{meanders}
\cite{Flajo,BanderierFlajolet02}. This way, we prove that $f(z)$ is an algebraic function,
characterised via the roots of a \emph{kernel} polynomial $K(z,u)$. Analysing $K(z,u)$, we can  under mild algebraic conditions  compute the radius of convergence $R$ of $f(z)$ and its singularities of minimal modulus (dominant):
this information directly translates into asymptotic formulae for $p_n$, as well as exact (non asymptotic)
exponential bounds \cite{Flajo}. Probability of termination, expectation and higher moments
can also be easily read off $f(z)$
as a byproduct.
The resulting method is sound, but not complete, for the class of GCP
programs; the conditions under which it succeeds are of linear-algebraic nature and easy to check. In
particular, completeness holds for the class of 1-state GCP, i.e.\ the CP format of \cite{Giesl19}.

In summary: our method provides precise information about the distribution of running time (exact asymptotics of coefficients,
exponential tail bounds) within a class of GCP programs admitting algebraic pgf's.
Despite the restrictions imposed by the GCP format, we give a complete symbolic analysis of challenging
programs, including the 
\textsf{ZeroConf} protocol.
We position our contribution as a bridge between program semantics and
analytic combinatorics, complementing existing proof techniques with a novel quantitative toolkit.
\begin{lstlisting}[caption={Running Example.},basicstyle=\ttfamily\small,label=lst:running_example,captionpos=b,float=h,
	abovecaptionskip=2pt,mathescape=true]
	while($d\geq0$){
	     if ($q=1$) then
		{$\frac{1}{2}$:$d:=d+1\,;\,q:=2$[]$\frac{1}{2}$:$d:=d-1\,;\,q:=1$}
	     else
		{$\frac{2}3$:$d:=d+1\,;\,q:=1$[]$\frac{1}6$:$d:=d-2\,;\,q:=2$[] $\frac{1}6$:$d:=-1$}}
\end{lstlisting}\label{ex:run}

\paragraph*{Related Work}
Giesl, Giesl and Hark~\cite{Giesl19} introduced CP programs and gave algorithms to compute
  tight bounds on expected runtime. We generalise their scope in two directions: (i) we
target a broader   class (GCP) of programs, whose semantics is given by algebraic pgf's; and (ii) we leverage
generating function techniques, permitting extraction of \emph{exact} asymptotic information on coefficients
and  tight exponential bounds.

Kaminski et al.~\cite{KaminskiKatoenMathejaOlmedo18} developed a general weakest-precondition (wp)
style calculus for reasoning about expected run-times of probabilistic programs. Ngo, Carbonneaux and
Hoffmann~\cite{CarbonneauxNgoHoffmann18} introduced \emph{Bounded Expectations}, a static analysis that
synthesizes symbolic polynomial bounds on expected resource consumption. Ranking supermartingales
and lexicographic extensions~\cite{AgrawalChatterjeeNovotny17,Chatterjee21,FioritiHermanns15} offer
general techniques for termination and quantitative bounds, but their output is usually limited to sound
bounds on expected value. Our approach is complementary, yielding
much more precise quantitative information whenever applicable.


More closely related to ours is the work  of Br\'{a}zdil et al.~\cite{BrazdilKieferKuceraVarekova15},
where they analyse pPDA runtime and provide exponential tail bounds via martingale techniques.
Our techniques yield much stronger bounds for programs where both formalisms are applicable, besides providing exact
asymptotics; we compare the two approaches in
Section \ref{sec:en}.

Esparza et al.~\cite{esparza} study expected values and higher moments of running time via linear
systems involving expected values as unknowns for  pPDAs.
Our approach is also somewhat related to work on termination and expected runtime in
probabilistic Vector Addition Systems with States~\cite{braz} and Probabilistic Counter Programs~\cite{pOCACM,AST};
and to a line of research by Kov\'{a}cs et al.~\cite{Kov1,Kov2,Kov3}
on the automated analysis of probabilistic loops, including higher moments. \highlight{While these approaches primarily address properties such as almost-sure termination or moments, our work provides asymptotic characterisations and tail bounds. In particular, whereas the line of work by Kov\'{a}cs et al. focuses on distributions over program state variables, we specifically study the distribution of the program runtime.}

The present investigation  is part of a broader research agenda, aimed at developing flexible formal methods applicable  across   diverse, probability-related domains,  including dynamical  systems with safety-related aspects
\cite{HSCC18,SOFSEM18,LMCS19,MFCS19,IC22,BorealeG21,VMCAI23}, information leakage and security   \cite{ISC14,QIFForte14,QIFLMCS15,BorealeCV20}, distributed  systems with notions of failure and recovery \cite{ABZ13,CaspisMSCS},   randomized    testing \cite{BG19,OutgenIEEE}.


%

\paragraph*{Structure of the paper} Section~\ref{sec:prelim} introduces preliminary notions on generating
and algebraic functions.   GCP programs
are introduced in Section~\ref{sec:GCPP}, together with their operational semantics via
a \highlight{\poc}-runs generating function $\Delta(z)$.
Our analysis of $\Delta(z)$  is split into two phases.  
In the first phase, presented in Section~\ref{sec:alg}, we adapt the \emph{kernel method} \cite{Flajo}
to compute an algebraic representation of $\Delta(z)$ (Algorithm \ref{alg:ana}).
This information is used as input in the second phase, presented in Section~\ref{sec:algan},
which  performs  a complete singularity analysis (Algorithm \ref{alg:radius}) leading to    asymptotics and exponential   bounds of $\Delta(z)$'s
coefficients .
Section~\ref{sec:en} introduces enhancements, discusses tail bounds, and further examples.
Concluding
remarks are drawn in Section~\ref{sec:concl}.
Detailed proofs and additional computational details \iffull have been confined to Appendices
\ref{app:proofs} and \ref{app:details}, respectively. \else   appear in a full version of the
paper available online \cite{Full}.\fi

\section{Preliminaries}\label{sec:prelim}

\subsection{Sub-probability generating functions}
We consider \emph{sub-probability generating functions (gf)} in the variable $z$. These
are functions of the complex variable $z$ of the form $f(z) = \sum_{n=0}^\infty a_n z^n$ such that
$a_n\in [0,1]$ and $f(1)= \sum_{n=0}^\infty a_n \leq 1$. If equality holds, we say $f$ is a
probability generating function (pgf), i.e.\ a distribution on $\N$. For $n\geq 0$, we let
$[z^n]f(z):=a_n$. Note that $f(1)=\sum_{n\geq 0} a_n$ is the total probability mass; $f'(1)=
\frac{\mathrm{d}}{\mathrm{d}z} f(z)_{|z=1}=\sum_{n\geq 1} n\cdot a_n$ is the expected value. Higher-order moments
can be computed similarly; see~\cite{Wilf05}.
Exponential   bounds for the coefficients are easy to obtain from knowledge of the radius of convergence $R>1$ of $f(z)$: for every $n\geq 0$ and   $r\in(1,R)$,
$a_n\leq f(r)\cdot(1/r)^n$ \cite[Prop.IV.1]{Flajo}; this bound can be easily converted
into a deviation-from-mean bound.

\subsection{Algebraic functions}
Let $K(z,u)\in \C[z,u]$ be a nonzero polynomial in $z$ and $u$ with complex coefficients. Let $\delta_u:=\deg_u(K(z,u))$ be the degree  in $u$ of $K$, i.e. $K(z,u)=p(z)u^{\delta_u}+O(u^{\delta_u})$, for some $p(z)\neq 0$. Let $ \Xi[K]:=\{z:\,\exists u\in \C \text{ s.t. } K(z,u)=0\text{ and } \partial K(z,u)/\partial u=0\}\cup\{z:p(z)=0\}$,  the \emph{exceptional set} of $K$.
For each   $z_0\notin \Xi[K]$,  the equation $K(z_0,u)=0$ has $\delta_u$   solutions $u_0\in \C$ (counted with multiplicities): this is a consequence of the Fundamental Theorem of Algebra, as   $p(z_0)\neq 0$ hence $K(z_0,u)\in \C[u]$ is a polynomial of degree $\delta_u$.  Moreover, for  each such solution $u_0$ the Implicit Function Theorem (IFT)   applies, as $K(z_0,u_0)=0$ hence $\frac{\partial}{\partial u}K(z_0,u_0)\neq 0$: this guarantees that there is a unique analytic (\replace{i.e.} admitting a Taylor series with positive radius of convergence) function $u(z)$  defined in a disk centered at $z_0$, $D_{z_0}=\{z: |z-z_0|<r\}$,   such that
$$
u(z_0)=u_0\quad \text{ and }\quad K(z,u(z))=0\text{ identically in $D_{z_0}$.}
$$
We call such $\delta_u$   analytic functions the \emph{branches} of $K(z,u)$ at $z_0$. Note that the IFT only gives a sufficient condition for the existence of an analytic solution of $K(z,u)$ near $z=z_0$: there might be points   $z_0\in\Xi[K]$ that nevertheless admit one or more analytic branches (see example below). To sum up,  except for a set   of points contained  in  $\Xi[K]$, at   each $z_0\in \C$  there are $\delta_u$   analytic branches of $K(z,u)$, counted with the same multiplicities as the roots of $K(z_0,u)$.

\newcommand{\Puiz}{\C\{\!\{z\}\!\}}
\newcommand{\Pui}{\C\{\!\{t\}\!\}}
\newcommand{\FPui}{\C\langle\!\langle t\rangle\!\rangle}
\newcommand{\Arg}{\mathrm{Arg}}
\newcommand{\val}{\mathrm{val}}
\newcommand{\ram}{\mathrm{ri}}
\newcommand{\ri}{\mathrm{ri}}
At points $z_0\in \Xi[K]$, the IFT does not apply. The standard way of describing what happens at these points, and more generally at \emph{any} point in $\C$,  is  appealing to the more advanced Puiseux theorem. Consider   the set  of formal power series in the variable $t$, with fractional exponents, of the form
$$S(t)=\sum_{n\geq n_0}c_nt^{n/\kappa}$$
 for $c_n\in \C$ and   integers $n_0\in \Z$ and  $\kappa>0$.
These are called \emph{formal Puiseux series} and denoted by $\FPui$. For a nonzero $S(t)$  with $c_{n_0}\neq 0$,  the quantity $\val(S):=n_0/\kappa$ is called \emph{valuation} of $S(t)$, while its \emph{ramification index} is $\ram(S):=\min\{e\geq 1: $ for each $n\geq n_0$ with $c_n\neq 0$, $n/\kappa$ is an integer multiple of $1/e\}$. One can choose $\kappa$ s.t. $\ram(S)=\kappa$.  Let us focus on the series $S(t)$ such that $\sum_{n\geq n_0}c_ns^{n}$  is \emph{convergent} for all $0<|s|<R$, for some $R>0$, and denote this set by $\Pui$.

\begin{theorem}[Puiseux theorem, \cite{Wall2004}, Th.2.4.3]\label{th:puiseux}
Both $\FPui$ and $\Pui$ are    algebraically closed fields.
\end{theorem}

Each     $S(t)\in \Pui$ can  be identified with a  complex-valued  function  once we fix a determination\footnote{Specifically, we fix a real $\phi\neq 0$ and    a complex argument function $\Arg_\phi:\C\setminus \{0\}\to (\phi,\phi+2\pi]$, then let $t^{1/\kappa}:= |t|^{1/\kappa}e^{{\mathlarger{\iota}}\Arg_\phi(t)}$, where $|t|^{1/k}$ is the usual real $k$-th square root. This function is   analytic  in   $\C\setminus \rho$, for $\rho$ the ray defined in the text. A common choice is $\phi=-\pi$ (principal branch).}   for  the  complex root $t^{1/\kappa}$. In particular,  $S(t)$ is defined and analytic for $t$ in   the \emph{slit disk} of $0$ defined by  $D  \setminus \rho$, where $D=\{z:|z|<R\}$  and  $\rho=\{\, \lambda \mathrm{e}^{i\phi} : \lambda \ge 0  \}$, for \highlight{some} fixed real nonzero $\phi$, is a \emph{ray} emanating from $0$.
Now consider again $K(z,u)$ and   fix any $z_0\in \C$.    Consider the change of variable $t=z-z_0$, i.e. the (translated) polynomial    $K(t+z_0,u)$: this  is also    a polynomial in $u$ with coefficients in $\Pui$, i.e. as an element of $\Pui[u]$. Theorem \ref{th:puiseux} guarantees that there are $\delta_u$ roots of $K(t+z_0,u)$, counted with multiplicities. Each of these roots $S(t)$ satisfies by definition $K(t+z_0,S(t))=0$ identically. In terms of the original variable $z$, we have   $K(z,S(z-z_0))=0$.   In other words,  at $z_0$ there are $\delta_u$  (again, counted with multiplicities) branches $u(z)=S(z-z_0)$ of $K(z,u)$, i.e. a Puiseux series in $z-z_0$ of the form
\begin{equation}\label{eq:PS}
u(z)=\sum_{n\geq n_0} c_n (z-z_0)^{n/\kappa}
\end{equation}
s.t. $K(z,u(z))=0$ identically.
The denominator $\kappa$ can be chosen to be the same for all branches at $z_0$, possibly taking the $\mathrm{lcm}$ across different branches. Additionally, the fact that,  as a complex function of $t$,  $S(t)$ is analytic for $t$ in the slit disk $D\setminus \rho$ of $0$ implies $u(z)=S(z-z_0)$ is analytic for $z$ in the slit disk, centered at $z_0$, $(z_0+D)\setminus (z_0+\rho)$ of $z_0$, with $z_0+X$ denoting the image of $X$ under  $z\mapsto z_0+z$.    We will often say \emph{near} $z=z_0$ to mean: for all $z$ in some slit disk of $z_0$.  Note that $u(z)$ in \eqref{eq:PS} is analytic in a disk (rather than just  in a \emph{slit} disk) of $z_0$ if and only if   $\val(u)=n_0/\kappa \ge 0$ and whenever $\kappa\not\!|n$ then $c_n=0$, or equivalently $\ri(u)=1$. In this case, we will silently assume no slit on the disk $D$. This criterion will be useful in Section \ref{sec:algan}.  A sufficient condition for \emph{all} branches  at $z_0$  to be analytic is $z_0\notin \Xi[K]$, as explained above.
Given   $K(z,u)$ and $z_0$, the Puiseux series expansion \eqref{eq:PS} of all branches at $z_0$ can be computed via   the \emph{Newton's polygon algorithm}; see \cite[Chapter~7]{Flajo} and \cite[Ch.2]{Wall2004}. Implementations of this algorithm are available for instance in {\sc Maple}.

$K(z,u)$ is \emph{square-free} if $\gcd(K,\partial  K/\partial u)=1$, where here $\gcd$ is computed considering $K,\partial  K/\partial u$  as polynomials in $\C(z)[u]$, i.e. polynomials in $u$ with coefficients in $\C(z)$,  the field of fractions of polynomials in $z$. In this case, the $\delta_u$ branches promised by Theorem \ref{th:puiseux} are actually \emph{distinct}. Moreover, $\Xi[K]$ is   finite. For example, $K(z,u)=(u-z)^2$ is not square-free, as  $\gcd(K,\partial  K/\partial u)=u-z\neq 1 $; indeed, $u=z$ is a double root for every $z$, and $\Xi[K]=\mathbb{C}$.
If $K(z,u)$ is not square-free, one can consider $H(z,u):=K(z,u)/\gcd(K,\partial  K/\partial u)$, which is square-free and has the same roots as $K$, all with multiplicity 1.
It can be shown that $H(z,u)$ is  still an ordinary polynomial in $\C[z,u]$ (basically as a consequence of Gauss lemma on polynomials).
%
%
We say that a function  is \emph{algebraic} at $z_0$ if it is  a branch at $z_0$ of some nonzero polynomial $K(z,u)$. 

\begin{example}\label{ex:prelim-kernel}
	\highlight{
Consider 	$K(z,u)= u- \frac 1 2 zu^2-\frac 1 2 z$, a polynomial that will arise in connection with a random walk program   where $\pm 1$ increments are chosen with equal probability at each  iteration (cf. Example \ref{ex:run6}). As $\partial K(z,u)/\partial u = 1-  zu$,  the solutions $(z,u)$ of
$\{K(z,u)=0,\partial K(z,u)/\partial u=0\}$ are $\{(-1,-1),(1,1)\}$; moreover the leading term in $u$ vanishes when $z=0$: hence  $\Xi[K]=\{-1,0,1\}$. Since $K(z,u)$ is square-free,
at any  $z_0$  Puiseux theorem predicts $\delta_u=2$ distinct branches. We can compute the Puiseux series of these branches applying Newton's polygon algorithm. In particular, for each $z_0\in\Xi[K]$
  one finds the following expansions of the two branches,  with the help of  a computer algebra system\footnote{A {\sc Maple} worksheet coding the examples presented throughout the paper  is available online: \url{https://osf.io/3kfr2/overview?view_only=1e39537d4d44485e98d288922f5286d8}.}.
\begin{itemize}
\item   $\mathbf{z_0=-1}$ : $u_1(z)=-1-\sqrt{2}\,(z+1)^{1/2}-(z+1)
-\frac{3\sqrt{2}}{4}(z+1)^{3/2}
+O\!\left((z+1)^2\right)$ (non analytic), $u_2(z)=-1+\sqrt{2}\,(z+1)^{1/2}-(z+1)
+\frac{3\sqrt{2}}{4}(z+1)^{3/2}
+O\!\left((z+1)^2\right)$ (non analytic);

\item  $\mathbf{z_0=0}$: $u_1(z)=(1/2)z+(1/8)z^3+O(z^4)$ (analytic), $u_2(z)=2/z-(1/2)z-(1/8)z^3+O(z^4)$ (non analytic);

\item   $\mathbf{z_0=1}$ : $u_{1}(z)
=
1+{\mathlarger{\iota}}\sqrt{2}\,(z-1)^{1/2}-(z-1)
-\frac{3{\mathlarger{\iota}}\sqrt{2}}{4}(z-1)^{3/2}
+O\!\left((z-1)^2\right)$ (non analytic), $u_{2}(z)
=
1-{\mathlarger{\iota}}\sqrt{2}\,(z-1)^{1/2}-(z-1)
+\frac{3{\mathlarger{\iota}}\sqrt{2}}{4}(z-1)^{3/2}
+O\!\left((z-1)^2\right)$ (non analytic).
\end{itemize}
Once one fixes a determination for the complex square root $t^{1/2}$, the above branches can be identified with    functions analytic in a (slit) disk of the corresponding $z_0$. E.g. fixing $t^{1/2}=\sqrt{t}$, one has the slit $\rho=(-\infty,0]$. The analytic  branch $u_1(z)$ at $0$ is independent of this determination.  Since $K(z,u)$ has a low degree, these branches can also be found solving explicitly for $u$ $K(z,u)=0$.  For example, at $z_0=0$, one finds
$$u_{1,2}(z) = \dfrac{1\pm\sqrt{1-z^2}}{z}$$
where the analytic branch $u_1(z)$ corresponds to taking the minus sign on the right-hand side.
}
\end{example}

\section{Generalised Constant Probability Programs}\label{sec:GCPP}
\lstset{basicstyle=\ttfamily,mathescape}
Probabilistic programs~\cite{barthe2020foundations} extend ordinary programs by allowing probabilistic
choice.
$\pGC$~\cite{mciver2005abstraction}  focuses on discrete probability distributions and does not consider \textsf{observe()} statements.
We  omit treating nondeterminism at this stage.
In what follows, we introduce \emph{Generalised Constant Probability (GCP)} programs,
a subset of $\pGC$.


\paragraph{Syntax and operational semantics of GCP}
A GCP   program is built from integer variables
$d$ and $q$, encoding a counter and a state respectively, and consists of a single \textsf{while} loop
whose guard tests non-negativity of the counter $d$.
Inside the \textsf{while}  loop, the state variable dictates
which  state- and counter-updates can be probabilistically performed.
In what follows,
$p_1:P_1\mathtt{[\ ]}\cdots\mathtt{[\ ]}p_k:P_k$  will denote probabilistic choice between
programs $P_1,\dots,P_k$, whenever $(p_1,\dots,p_k)$ is a probability distribution.



\begin{definition}[GCP programs]\label{def:GCPP}
	Let $N \geq 1$ (number of \emph{states}) and $\TT =[(b_1,\J_1),\ldots,(b_N,\J_N)]$  (\emph{transition list})
	be such that  for each $i\in\{1,\dots,N\}$: $b_i\in[0,1]$, $\J_i\subseteq (0,1]\times
	\Z\times\{1, \dots,N\}$ is a finite set and $b_i+\sum_{(p,k,q')\in\J_i}p=1$.
	%
	The \emph{GCP program associated with} $\TT$ is:
	{\small{\em
			\begin{center}
				\begin{minipage}{\linewidth}
					\begin{lstlisting}
while ($d\geq 0$){
      if $q=1$   {$b_1$: $\;d:=-1$  $\mathtt{\left[\ \right]}$  $\;\mathtt{\left[\ \right]}_{(p,k,q')\in\J_1}$ $p:d:=d+k;\,q:=q'$}
      else if $q = 2$ {$b_2$: $\;d:=-1$  $\mathtt{\left[\ \right]}$  $\;\mathtt{\left[\ \right]}_{(p,k,q')\in\J_2}$ $p:d:=d+k;\,q:=q'$}
      $\scriptsize\vdots$
      else         {$b_N$: $d:=-1$   $\mathtt{\left[\ \right]}$  $\,\mathtt{\left[\ \right]}_{(p,k,q')\in\J_N}$ $p:d:=d+k;\,q:=q'$}  }.
					\end{lstlisting}
				\end{minipage}
	\end{center}}}\noindent
\end{definition}
For a given program $P$, it is natural to interpret   $\TT=\TT_P=[(b_1,\J_1),...,(b_N,\J_N)]$ as
the transitions of a  probabilistic \emph{Push-Down Automaton} (PDA)
with  a one-symbol stack alphabet say $\{u\}$, that is  a \poc. 
The number of elements  on the stack (\emph{height}) represents the value of $d$ at a given moment.
For $j\in\Z$, assignments $d:=d+j$ are \emph{relative
	updates}, and $d:=j$ are \emph{absolute updates}.
Relative updates correspond to pushing/popping symbols onto/from the stack.
Additional immediate termination
arcs  represent absolute updates; see Fig. \ref{fig:run1} (right).
We now develop this correspondence rigorously, by giving an operational semantics of GCP programs in terms of PDAs.

Fix any integer\footnote{This $d_0$ has the sole  purpose of representing  termination by   absolute updates, keeping it distinct from termination by negative relative updates.}
$d_0<0$,  less than every integer occurring in $\TT_P$.
Let us denote a    pair $(k,q)\in \Z\times \{1,...,N\}$ just as $kq$. A \emph{run}  of $P$ of \emph{length} $n\geq 0$ is a sequence of the form
\begin{align}\label{eq:run}
\rho & =k_0q_{0},\; p_1,\; k_1q_{1},\,p_2,\ldots,  \; p_{n},\; k_nq_{n}
\end{align}
such that for each $t=0,...,n-1$:    either  (a) $k_t\geq 0$, $(p_{t+1},j_{t+1},q_{{t+1}})\in \J_{q_{t}}$ and $k_{t+1}=k_{t}+j_{t+1}$; or (b) $t=n-1$, $q_{{n-1}}=q_{n}$, $p_n=b_{q_n}$ and $k_n=d_0$. Case (b)   corresponds to termination by an absolute  update. Note however that it may well be $k_n<0$ even if the last transition is a relative negative update. Also note that, for $n=0$,  $\rho=k_0q_{0}$ with $k_0<0$ is allowed. 

The \emph{length} of $\rho$ is  $|\rho|:=n$; for $n=0$, we have $\rho=k_0q_{0}$.
The \emph{initial} and \emph{final  state} of $\rho$ are $\is(\rho):=q_{{0}}$ and $\fs(\rho):=q_{{n}}$, respectively. Its initial and final \emph{altitude} are $\ia(\rho):=k_0$ and $\fa(\rho):=k_n$, respectively.  
The \emph{weight} of $\rho$ is: $\wt(\rho):=1$ if $n=0$, and $\wt(\rho):=p_1\cdots p_{n}$ if $n>0$.
\emph{Terminal} runs, ranged over by $\tau$, \replace{ are those where} $k_n<0$. \emph{Non-terminal} runs are also called \emph{meanders} and are ranged over by $\sigma$.
%
For $d\in \Z$, $n\geq 0$ and $i=1,...,N$, let us define the following
quantities, giving the global weight of terminal runs of length $n$  starting from state $i$   and altitude  $d$:
\begin{align}\label{eq:terminal}
D^d_{i,n}&:=\sum\{\wt (\tau)\,:\,|\tau|=n,\,\ia(\tau)=d,\,\is(\tau)=i  \}\,.
\end{align}
(note that whenever $d<0$, we have $D^d_{i,n}=[n=0]$, while for $d\geq 0$ we have $D^d_{i,0}=0$).

\begin{definition}[operational gf]\label{def:Delta}
The \emph{operational} gf of initial state $i\in 1..N$ and counter value $d\in \Z$ is
$\Delta^d_i(z) := \sum_{n\geq 0} D^d_{i,n}z^n$.
\end{definition}

Thus $[z^n]\Delta^d_i(z)$ is the probability that $P$, starting from counter value $d$ and state $i$,
terminates in exactly $n$ iterations. $\Delta^d_i(1)$ is the total probability of termination, and
$\frac{\mathrm{d}}{\mathrm{d}z}\Delta^d_i(z)_{|z=1}/\Delta^d_i(1)$ is the expected number of steps to
termination given termination occurs.

\paragraph{A meander-based characterisation}
We develop now a crucial  meander-based characterisation of $\Delta^d_q(z)$.
To avoid proliferation of indices, we  study here only $\Delta^d:=\Delta^d_1$, that is the terminal runs starting from the (initial)
state $q=1$. For $d,n,h\geq 0$ and $j=1,...,N$, let us define the following   quantities,
relative to meanders \textbf{starting in state} $1$ with \textbf{initial altitude} $d$ and
{\textbf{ending in state} $j$} with \textbf{final altitude}  {$h\geq 0$}:
{
\begin{align*}
	m^d_{j,n,h}&:=\sum\{\wt(\sigma):|\sigma|=n,\, \ia(\sigma)=d, \, \is(\sigma)=1,\,\fa(\sigma)=h,\,\fs(\sigma)=j \}\\
	m^d_{j,n}&:=\sum_{h\geq 0} m^d_{j,n,h}.
\end{align*}
}
\begin{definition}[meanders gf's]\label{def:gfmeanders} For each $1\leq j\leq N,\,d\in \Z$,
we define the following
gf of meanders starting in state $1$ at altitude $d$ and ending in state $j$.
Variable $z$ marks the length, while   $u$ marks the final altitude of the meanders.
\setlength{\arraycolsep}{0.3pt}
$$
\begin{array}{rcll}
	M^d_{j}(z,u)&:=\sum_{n,h\geq 0}  m^d_{j,n,h} \cdot z^n u^h  \hspace{1.5cm} \text{(general  bivariate gf)}\\
	M^d_{j,h}(z)&:=\sum_{n \geq 0}  m^d_{j,n,h} \cdot z^n \hspace{0.5cm} \text{(ending with final altitude $h\geq 0$)}\\
	M^d_{j}(z)&:=\sum_{n \geq 0} m^d_{j,n}\cdot z^n \hspace{0.2cm} \text{(ending at \emph{any} nonnegative altitude).}
\end{array}
$$
\end{definition}

Note that $M^d_{j}(z)=M^d_{j}(z,1)$. A terminal run $\tau$ can be decomposed as $\tau=\sigma,\eta$, where $\sigma$ is a meander and $\eta=(p_n,\ k_{i_n}q_{n})$ a single negative update. This  gives rise to a decomposition of $\Delta^d$ in terms of the meanders gf's. We illustrate the underlying idea with an example, before stating the general result.

\begin{example}{\label{ex:run3}
	Consider our running example.  Assume an initial altitude $d\geq 0$. With an eye to the automaton in Fig. \ref{fig:run1} (right), it is easy to see that  terminal runs can be formed starting from  meanders $\sigma$  that:  (a) arrive   at state 1 with altitude 0 and are followed by a   relative update of -1 with probability $1/2$ (accounted for by $\frac 1 2 z\cdot M^d_{{1,0}}(z) $); (b)    arrive  at state 2 with altitude 0 or 1 and are  followed in each case by a  relative update of $-2$  with probability $1/6$: (accounted for by $\frac 1 6 z\cdot M^d_{{2,0}}(z)+\frac 1 6 z\cdot M^d_{{2,1}}(z)$);  (c) arrive  at state 2 with \emph{any} altitude and are followed by a  negative \emph{absolute} update  with probability $1/6$ (accounted for by $\frac 1 6 z\cdot M^d_{{2}}( z  ) $).
	{\begin{align}\label{eq:run3}
			\Delta^d(z)&=z \left(\;  \frac 1 2   M^d_{{1,0}}(z)   + \frac 1 6
			M^d_{{2,0}}( z  ) + \frac  1 6    M^d_{{2,1}}( z ) + \frac  1  6 M^d_{{2}}( z  )  \;\right).
		\end{align}
	}\noindent
}\end{example}

In general, consider the transition list $\TT_P$ and let $e_j$ be the absolute value of the minimal \emph{negative} relative update occurring in $\J_j$ (0 if none occurs),
and $p_j^{<h}$ the sum of the probabilities of all relative updates $<h$ in $\J_j$ (0 if none occurs).
Below, we adopt the Iverson bracket notation: $[\mathcal{P}] = 1$ if the predicate $\mathcal{P}$ is true,  $0$ otherwise.

\begin{theorem}\label{th:char2}For $d\in \Z$:
{
	$\Delta^d(z) =[d<0] 1+[d\geq 0] z\cdot\sum_{j=1}^N \left(\sum_{h=0}^{\replace{e_j-1}} {p_j^{<-h}} M^d_{j,h}(z)\;+  b_j M^d_j(z) \;\right).
	$}\noindent
\end{theorem}


\usetikzlibrary{automata, positioning}
\begin{minipage}{\columnwidth}
\begin{minipage}[c]{0.6\textwidth}
	\begin{lstlisting}[mathescape=true,basicstyle=\ttfamily\footnotesize]
while($d\geq0$){
     if ($q=1$)
        {$\frac12$:$d:=d+1\,;\,q:=2$ [] $\frac12$:$d:=d-1\,;\,q:=1$}
     else
        {$\frac23$:$d:=d+1\,;\,q:=1$ [] $\frac16$:$d:=d-2\,;\,q:=2$ [] $\frac16$:$d:=-1$}
}
	\end{lstlisting}
\end{minipage}
\begin{minipage}[c]{0.37\textwidth}
		\begin{tikzpicture}[->, >=Stealth, auto, node distance=4.5cm, semithick,every node/.style={scale=0.9, text=black}, thick, draw=black]
			\tikzstyle{every state}=[fill=white, draw=black, text=black, minimum size=1pt, draw=black]
		\node[state] (q1) {1};
		\node[state, right of=q1] (q2) {2};
		\path (q1) edge[bend left] node[above] {$\frac{1}{2}, u^1$} (q2);
		\path (q1) edge[loop above] node {$\frac{1}{2}, u^{-1}$} (q1);
		\path (q2) edge[bend left] node[below] {$\frac{2}{3}, u^1$} (q1);
		\path (q2) edge[loop above] node {$\frac{1}{6}, u^{-2}$} (q2);
		\draw[dashed] (q2) -- ++(1,0) node[midway, above] {$\frac{1}{6}$};
	\end{tikzpicture}
\end{minipage}

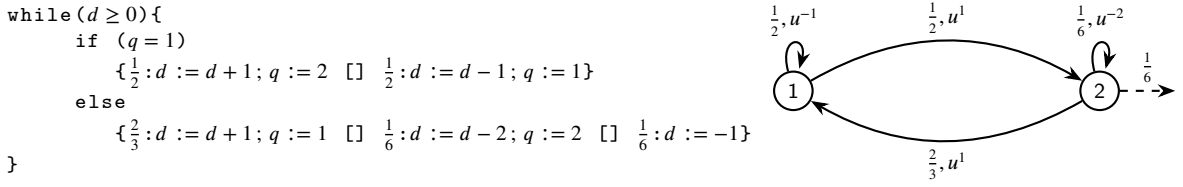
\captionof{figure}{\small\textbf{Left}: Running example. \textbf{Right}: its associated \highlight{\poc}. A	solid transition labeled $p,u^j$ occurs with probability $p$ and pushes ($j\ge 0$) or pops ($j<0$) $|j|$ 	elements into/from the stack. A dashed transition, or attempting to pop more than the stack contains, leads to termination.}\label{fig:run1}
\end{minipage}



\section{The Kernel Method for GCP Programs}\label{sec:alg}
The main outcome of this section is Proposition~\ref{prop:rational}, basically saying how to   compute  an algebraic representation of $\Delta^d(z)$.  In view of Theorem \ref{th:char2}, this is achieved by introducing  a linear system  satisfied by \emph{meanders} gf's, and solving it under a linear-algebraic rank condition. This is the essence of the \emph{kernel method},  which we adapt from Banderier and Flajolet\replace{'s} work, see \cite{Flajo,BanderierFlajolet02,AsinowskiBacherBanderierGittenberger20}. The resulting representation will then be used in the subsequent analysis of $\Delta^d(z)$ in Section \ref{sec:algan}.

The key building blocks are the (Laurent) polynomials
$$S_{ij}(u):=\sum_{(p,k,i)\in\J_j}pu^k$$
that account for all transitions from $j$ to $i$, tracking probability $p$ and update $k$.
%
If our program   had only \emph{nonnegative} relative updates, the  meanders bivariate  gf's would obey the equalities:
$$M^d_i(z,u)=[i=1] u^d+z\sum_{j=1}^N S_{ij}(u)M^d_j(z,u) \ \  \text{ for $i=1,...,N$}$$
(NB: the integer `$d$' in $u^d$ is an actual exponent, not a superscript). In the presence of \emph{negative} relative updates, this equality breaks down:  e.g. $(p,-2,i)\in \J_j$ introduces in the RHS   a   term  $z\cdot pu^{-2}\cdot M^d_{j}(z,u)$ that   contains spurious summands. Indeed,  $M^d_{j}(z,u)=M^d_{j,0}(z)+u\cdot M^d_{j,1}(z)+u^2\cdot M^d_{j,2}(z)+\cdots$, so  $z\cdot pu^{-2}\cdot M^d_{j}(z,u)=zpu^{-2}M^d_{j,0}(z)+zpu^{-1}\cdot M^d_{j,1}(z)+zpu^0\cdot M^d_{j,2}(z)+\cdots$ contains negative powers of $u$, not present in the LHS. On a \highlight{\poc}  (Fig. \ref{fig:run1}, right), those terms correspond to   transitions not allowed for meanders:  popping  from the stack more symbols than available, thus ending up in negative heights $-2$ and $-1$, respectively. To keep  the balance with   $M^d_i(z,u)$, these terms with negative powers of $u$, that is $zpu^{-2}M^d_{j,0}(z)+zpu^{-1}\cdot M^d_{j,1}(z)$,   must be subtracted  from the RHS.  Systematically,  this   corresponds to subtracting the balancing expression  $B_i$:
{
\begin{align}
	B_{i}(z,u)&:=z\sum_{j=1}^N \sum_{\ (p,k,i)\in \J_j:k<0} p\sum_{h=0}^{-k-1} u^{k+h}\cdot M^d_{j,h}(z)\label{eq:Bi}\,
\end{align}
}\noindent
($B_i$ depends on $d$, although this is not made explicit in the notation).

\begin{lemma}\label{prop:Meqs}
For $i=1,\ldots,N$: $M^d_i(z,u) \,=\,[i=1]u^d+z\sum_{j=1}^N S_{ij}(u)M^d_j(z,u)-B_i(z,u)$.
	\end{lemma}
	In vector-matrix notation, letting $M^d:=(M^d_1,\ldots,M^d_N)^T$, $B:=(B_1,\ldots,B_N)^T$, $A(u)$
	the $N\times N$ matrix of elements $S_{ij}(u)$, and $v^d(u):=(u^d,0,\ldots,0)^T$:
	\begin{align}\label{eq:Meq0}
		M^d(z,u)&=v^d(u)+zA(u)\cdot M^d(z,u)-B(z,u)\,.
	\end{align}
	Equation \eqref{eq:Meq0} leads  to
	\begin{align}\label{eq:Meq}
		(I-zA(u))\cdot M^d(z,u)&=v^d(u)-B(z,u)\,.
	\end{align}
	Left-multiplying by $w\cdot\Adj(I-zA(u))$, where $w=(1,\ldots,1)$ and $\Adj(I-zA(u))$ is the \emph{adjoint} matrix, (i.e. the  $N\times N$  matrix  satisfying\footnote{Explicitly, for a matrix $E$, $\Adj(E)[k,\ell]=(-1)^{k+\ell}\det{(E_{\ell k})}$, where $E_{\ell k}$ is the $(N-1)\times(N-1)$ matrix obtained from $E$ by removing row $\ell$ and column $k$.} $\Adj(I-zA(u))\cdot (I-zA(u))=\det(I-zA(u))\cdot I$),  we obtain
	{\begin{align}\label{eq:MainEq}
			\det(I-zA(u))\,w\cdot M^d(z,u) &= w\cdot\Adj(I-zA(u))\cdot(v^d(u)-B(z,u))=:L(z,u,\widetilde{M}^d_{j,h}(z))
	\end{align}}
	where in $L$ we have collected all $M^d_{j,h}(z)$'s into a tuple denoted $\widetilde{M}^d_{j,h}(z)$.
	
	\begin{definition}[kernel polynomial]\label{def:kernelpol}
		Let $e$ be the absolute value of the minimal negative degree of $u$ in $\det(I-zA(u))$ (0 if none) and let $\widetilde K(z,u):=u^e\cdot\det(I-zA(u))$. The
		\emph{kernel polynomial} is the square-free part of $\widetilde K$, i.e. $K(z,u):=\widetilde K(z,u)/\gcd(\widetilde K,\partial \widetilde K/\partial u)$.
	\end{definition}
	
	$K(z,u)$ is an   ordinary polynomial, which is by construction nonzero (e.g. $K(0,u)\neq 0$) and square-free. In particular, $\widetilde K(z,u)$ and $K(z,u)$ have the same roots as polynomials in $\Puiz[u]$, but  their multiplicities in $K$ is 1.   Near $z=0$, there are $\delta_u=\deg_u(K)$ distinct
	branches $u(z)$ with $K(z,u(z))=0$. A branch $u_s(z)$ is \emph{small} if it is nonzero 
and	$|u_s(z)|\to 0$ as $|z|\to 0$; equivalently, the Puiseux series of $u(z)$ has valuation $>0$.
Let $c$ denote the number of distinct small branches near $z=0$. For a small branch $u_s(z)$, we have $\widetilde K(z,u_s(z))=0$, hence $\det(I-zA(u_s(z)))=0$. Therefore, substituting\footnote{Substitution into $M^d(z,u)$ is well-defined because $M^d(z,u)$ is analytic in  $(0,0)$ and the $u_s(z)$ are small.} $u$ with $u_s(z)$  into~\eqref{eq:MainEq} makes  the LHS vanish. Doing so for each  small branch in turn, we obtain the
	system:
	\begin{equation}\label{eq:system}
		L(z,u_s(z),\widetilde{M}^d_{j,h}(z))=0,\quad s=1,\ldots,c\,.
	\end{equation}\noindent
	The $  \widetilde{M}^d_{j,h}(z)$'s  occur  \emph{linearly} in \eqref{eq:system}, so this is actually an affine system in the unknowns $ {M}^d_{j,h}(z)$'s with coefficients in $\Puiz$. If the rank of this system is $\geq|\widetilde{M}^d_{j,h}(z)|$, it has a unique solution.
	Once the $M^d_{j,h}(z)$'s are known, we solve~\eqref{eq:Meq} for $M^d(z,u)$, then set $M^d_j(z)=M^d_j(z,1)$. An explicit formula $Q$ for $\Delta^d(z)$ is then obtained via Theorem \ref{th:char2}. The following proposition additionally points out that  $Q$ is rational and can be computed \emph{symbolically}.
	
	\begin{proposition}\label{prop:rational}
		Assume the rank of~\eqref{eq:system} is $\geq|\widetilde{M}^d_{j,h}(z)|$. Then for each $d\geq 0$
		there is a rational expression $Q(z,u_1,\ldots,u_c)=\Num/D$ (with $\Num,D\neq 0$ polynomials) such
		that $\Delta^d(z)=Q(z,u_1(z),\ldots,u_c(z))$ near $z=0$. Moreover, $Q$ can be effectively computed
		given only $P$ and $d$, working in a fraction field with indeterminates $z,u_1,\ldots,u_c$.
	\end{proposition}
	
	
	\begin{example}\label{ex:run6}
		\highlight{For the purpose of illustration, we first consider a simple 1-state program,
modelling a random walk.}
\begin{center}
			$P:$ \ \ \ \texttt{while ($d\geq 0$)\{$\frac 1 2$:$d:=d+1$[]$\frac 1 2$:$d:=d-1$\}}.
\end{center}
\noindent
Theorem \ref{th:char2} gives $\Delta^d(z)=(1/2) z M^d_{1, 0}(z)$. We have $A(u)=S_{11}(u)=\frac 1 2 u+\frac 1 2 u^{-1}$, hence $e=1$, and the kernel polynomial is
$$K(z,u)=u(1-zA(u))=u(1-z((1/2) u+(1/2) u^{-1}))= u- (1/2) zu^2-(1/2) z$$
(this is already square-free). We have seen in Example \ref{ex:prelim-kernel} that this polynomial has one small branch $u_1(z)$ at $z=0$, that is $c=1$  (this small branch also happens to be analytic).
 Hence, for any $d\geq 0$,  \eqref{eq:system} reduces to the single equation
 $$u_1^d(z)-(1/2) z M^d_{1, 0}(z)/u_1(z) = 0$$
which can be  solved for $ M^d_{1, 0}(z)$, yielding
$$ M^d_{1, 0}(z)=2 u^{d+1}_1(z)/z.$$
From the previous expression for $\Delta^d(z)$, we obtain
$$\Delta^d(z)=Q(z,u_1(z)):=u_1^{d+1}(z).$$
At this point, for any fixed value of $d$, we can already expand directly $u^{d+1}_1(z)$ into   Puiseux series via the Newton's polygon algorithm. E.g. for $d=0$ we find, again in agreement with Example \ref{ex:prelim-kernel}:			
$$\Delta^0(z)=u_1(z)=\frac 1 2 z+ \frac 1 8 z^3+ \frac 1{16}z^5+O(z^7).$$
In this simple case, as seen from Example \ref{ex:prelim-kernel}    we can solve  $K(z,u)=0$ explicitly for $u_1(z)$ and find 
$u_1(z)= \frac{1-\sqrt{1-z^2}} z$.
The same method can be applied to our running example. In this case we have the kernel polynomial
$$K(z,u) = -\frac1{12}(4 u^5 z^2 - 12 u^3 + 6 u^2 z + 2 u z - z^2)$$
which has  $c=3$ small branches at $z=0$.
   Proceeding as we did above and exploiting the meander decomposition in \eqref{eq:run3} (see Appendix \ref{app:details}.1 for details on the computation), we find the Puiseux series: $$\Delta^0(z)=\frac 1 2 z+\frac  1 6{z}^{2}+\frac 1{36}{z}^{4}+ \frac {7}{108}{z}^{5}+\frac  1{36}\,{z}^{6}
				+O ( {z}^{7} ).$$
	\end{example}

\vsp
We summarize the workflow of  our algebraic analysis  as Algorithm \ref{alg:ana}.
The algorithm  fails at  step 3 if the rank of \eqref{eq:system}  is  $<|\widetilde{M}^d_{j,h}\replace{(z)}|$; in particular, this is the case if $c<|\widetilde{M}^d_{j,h}\replace{(z)}|$ (not enough small branches). Even if $c=|\widetilde{M}^d_{j,h}\replace{(z)}|$,   we might   still  have that the square matrix associated to  \eqref{eq:system} is singular. The determinant of this matrix is in turn an algebraic function, for which complete zero-tests exist \cite{zippel}. We omit the details of this aspect; in a practical implementation, an heuristic --although in general unsound-- alternative is computing a Puiseux series of the function in question up to a  fixed high order, and seeing if any nonzero coefficient arises.
Algorithm \ref{alg:ana}
never fails for nontrivial 1-state programs. We say  $P$ is \emph{nontrivial} if it has at least one negative relative update. A trivial 1-state program is uninteresting, in that it either never terminates (in case of no absolute updates), or it gives rise to a geometric distribution of parameter $\lambda$ (absolute update of probability $\lambda$).

\begin{proposition}\label{prop:compl} If Algorithm~\ref{alg:ana} does not raise \textbf{Fail}  then it returns $K$ and $Q$ as specified by \highlight{Def.} \ref{def:kernelpol} and Prop. \ref{prop:rational}, respectively.
Moreover, if $P$ is a nontrivial 1-state GCP program, Algorithm~\ref{alg:ana} does not raise
\textbf{Fail}.
\end{proposition}

{
\begin{algorithm}[t]
\caption{Computation of $K$ and $Q$}\label{alg:ana}
\begin{algorithmic}[1]
\Statex\textbf{Input:} $P$, a GCP  program; $d\geq 0$, an  (initial altitude)
\Statex\nonumber\textbf{Output:}  a nonzero, square-free $K$, and a rational function $Q=T/D$.
\State Compute the kernel polynomial $K(z,u)$ as in Definition \ref{def:kernelpol}
and build linear system \eqref{eq:system}.\label{line:l1}
\State	Using the Newton's polygon algorithm, compute  all the Puiseux series of  $K(z,u)$ at $z=0$, and identify  among them the $c$  {small branches} as those that are nonzero and have   valuation $>0$. 
\State If \eqref{eq:system} has rank  $<|\widetilde{M}^d_{j,h}(z)|$  then {\textbf{Fail}}; otherwise assume w.l.o.g. $|\widetilde{M}^d_{j,h}(z)|=c$.\label{line:l2}
\State Solve symbolically: \eqref{eq:system} for  $\widetilde{M}^d_{j,h}(z)$, then   \eqref{eq:Meq} for $M^d(z,u)$, and let $M^d_j(z):=M^d_j(z,1)$.\label{line:l3}
\State Combine $\widetilde{M}^d_{j,h}(z),  \ M^d_j(z)$ to obtain a rational expression $Q(z,u_1,...,u_c)=\Num(z,u_1,...,u_c)/D(z,u_1,...,u_c)$ s.t.  $\Delta^d(z)=Q(z,u_1(z),\dots,u_c(z))$ (Theorem \ref{th:char2}, Proposition \ref{prop:rational}).\label{line:l4}				
\State \textbf{return}  $K(z,u)$, $Q(z,u_1,...,u_c)$. \label{line:l5}
\end{algorithmic}
\end{algorithm}}

\section{Singularity analysis of $\Delta^d(z)$}\label{sec:algan}
\newcommand{\es}{\emptyset}
We   introduce an algebraic analysis method of $\Delta^d(z)$, which takes as input the rational expression $Q(z,u_1,...,u_c)$  and the kernel polynomial $K(z,u)$ returned by Algorithm \ref{alg:ana}. Under mild conditions, the method   returns an explicit  formula in $n$ of  the asymptotic growth of the coefficient $c_n$ of  $\Delta^d(z)$, as well as exact exponential upper bounds.  Computation of $R$ is essential for this analysis.
We begin by making the concept of singularity of a function explicit.

\begin{definition}[singularity \cite{Flajo}, \textbf{Def.IV.4}]\label{def:sing}
Let $f$ be a function which is defined and analytic in the interior region $\Omega$ of a simple closed curve\footnote{A  simple closed curve  $C$ (e.g. a circle) is the image of a continuous function $\phi:[0,1]\to \C$ s.t. $\phi(0)=\phi(1)$ and $\phi$ is injective on $[0,1)$. Any such curve divides $\C\setminus C$ into two connected components, the interior and the exterior, such that $C$ is the boundary (frontier) of both (Jordan theorem).} 
in $\C$. Let $z_0$ be a point in the boundary of $\Omega$. We say $f$ \emph{can be analytically continued to} $z_0$ if there is an analytic function $f^*$ defined in a disk $B$ of $z_0$ such that $f=f^*$ in $\Omega\cap B$. If $f$ cannot be analytically continued to $z_0$ we say $z_0$ is a \emph{singularity} of $f$.
\end{definition}

When $f$ can be analytically continued   to $\zeta$, we will call the function $f\cup f^*_{|U}$ (identifying functions with their graphs), where $f^*$ and $U$ are chosen like in Definition \ref{def:sing}, an \emph{analytic continuation} of $f$ to $\zeta$.
We recall some well-known facts on generating functions and   their singularities, see \cite[Prop.IV.1,Th.IV.7]{Flajo}.

\begin{lemma}\label{lemma:R} Let $f(z)$ be an algebraic function analytic at $z=0$, and $R$ the radius of convergence of its Taylor series, say $f(z)=\sum_{n\geq 0}c_n z^n$. Then $f(z)$ is analytic for $z$ in the disk $\{z:|z|<R\}$. Additionally
\begin{enumerate}
\item $R=\min\{|\zeta|\,:\,\zeta\text{ is a singularity of }f \}$, with $\min\es:=+\infty$.
\item Assume the  $c_n$     are all nonnegative. Then: (a)  for each  $r\in (0,R)$ and $n\geq 0$: $c_n\leq f(r)/r^n$;  (b) If $R<+\infty$ then $R$ is a singularity of $f$ (Pringsheim theorem).
\end{enumerate}
\end{lemma}

A singularity of minimal modulus of an analytic function at $z=0$ is called \emph{dominant}. According to the previous lemma, part 1,  computing $R$ is equivalent to locating a dominant singularity.  Moreover, for generating functions with nonnegative coefficients, like probability g.f.'s, Lemma \ref{lemma:R}(2,b) ensures that   a dominant singularity on $[0,+\infty)$ must exist, if there are any. The crucial result for the computation of $R$  is   the next theorem. We need a few  definitions and facts from analytic combinatorics \cite[Ch.VI,VII]{Flajo}.
In what follows, we convene that   slits from a point $\zeta\neq 0$ are taken along a ray $\rho$ not intersecting the disk $D_\zeta:=\{z:|z|<|\zeta|\}$, e.g. $\rho=\{\lambda\zeta:\lambda \geq 1\}$.
For distinct $z_0,z_1\in \C$, a \emph{path} connecting $z_0$ to $z_1$ is an injective continuous function  $\gamma:[0,1)\rightarrow \C$  such that $\gamma(0)=z_0$ and $\lim_{t\to 1^-} \gamma(t)=z_1$. Let $H$ be a square-free polynomial and $z_0\notin \Xi[H]$: then     provided $\gamma$ does not intersect $\Xi[H]$ and the fixed slit ray  at $z_1$, any (analytic) branch $u$ of $H$ at $z_0$ can be analytically continued to any point $\zeta \in \gamma([0,1))$ \cite[Lemma VII.4]{Flajo}. We will denote such continuation by the same name/symbol $u$. This justifies the following definition.

\begin{definition}[branches connected along a path]\label{def:conn} Let $H(z,u)$ be a square-free polynomial and $\zeta\in \C\setminus \{0\}$.  Let  $u(z)$  be  a    branch at $z=0$, and $v(z)$  be a branch at $z=\zeta$.  Let $z_0\notin\Xi[H]$ be a point near 0  where $u(z)$ is defined and analytic, and  $\gamma:[0,1)\rightarrow D_\zeta\setminus \Xi[H] $  be a path connecting $z_0$ to $\zeta$. We say that $v(z)$ is a \emph{branch (of H) at $\zeta$ analytically connected to $u(z)$ along} $\gamma$ if $u(\gamma(t))=v(\gamma(t))$ for all  $t\in (1-\epsilon,1)$, for some $0<\epsilon<1$.
\end{definition}

\newcommand{\AC}{\mathrm{AC}}
\begin{lemma}\label{lemma:AC} Assume $H$ is square-free, and $\zeta$ and $\gamma$ as in Definition \ref{def:conn}. Then, for each branch $u(z)$ of $H$ at $z=0$ there is a unique branch $v(z)$ of $H$ at $z=\zeta$, denoted as $v(z) =\AC(H,\zeta,\gamma,u)$,   that is analytically connected to $u(z)$ along $\gamma$. If additionally $u(z)$ is analytic at $z=0$ with radius of convergence $r$ and $|\zeta|\leq r$   then $v(z)$ does not depend on $\gamma$.
\end{lemma}

$\AC(H,\zeta,\gamma,u)$ can be identified with its Puiseux series in $z-\zeta$, once the slit ray from $\zeta$ has been fixed according to the convention illustrated above. In the particular case when $u(z)$ is analytic at $z=0$ and   $|\zeta|\leq r$ (independence from $\gamma$), we will call this simply \emph{the Puiseux series of $u(z)$ at} $z=\zeta$.
The  theorem below relies on standard arguments from the theory of algebraic functions.

\begin{theorem}[locating singularities]\label{th:singAng} Let $f$ be  an algebraic function analytic at $z=0$, 
and $R>0$ be the radius of convergence  of its Taylor series. 
Let $Q(z,u_1,...,u_c)$ be a rational function of its arguments, $K(z,u)$ be a square-free polynomial and  $u_1(z),...,u_c(z)$ be branches of $K(z,u)$ at $z=0$. Assume
$$f(z)=Q(z,u_1(z),...,u_c(z))\ \text{ near }\ z=0.
 $$
Fix $\zeta\in \C\setminus\{0\}$ with $|\zeta|\leq R$ and let $S_i:=\AC(K,\zeta,\gamma,u_i)$  ($i=1,...,c$, for one and the same path $\gamma$ as per Def. \ref{def:conn}), identified with its Puiseux series in $z-\zeta$.
Define $S_\zeta(z):=Q(z,S_1(z),...,S_c(z))$ as an element of the field of   Puiseux series in $(z-\zeta)$. Then $S_\zeta$ is the Puiseux series of $f$ at $z=\zeta$. Additionally:
\begin{enumerate}
\item $\zeta$ is a singularity of $f(z)$ if and only if $S_\zeta(z)=\sum_{n\geq n_0}c_n(z-\zeta)^{n/\kappa}$  is not a Taylor series (i.e. iff  $\val(S_\zeta)<0$ or $\ri(S_\zeta)>1$).
\item $L=\lim_{z\rightarrow \zeta} f(z)$   exists when taken along the chosen path $\gamma$. $L$ is finite if and only if $n_0\geq 0$. In particular, $L=c_{n_0}$ if $n_0=0$, $L=0$ if $n_0>0$.
\end{enumerate}
\end{theorem}

Part 1 of the above theorem  gives a necessary and sufficient   condition for locating the singularities of $f(z)=\Delta^d(z)$, based on $S_\zeta(z)$.
Armed with this criterion, we can compute  the radius of convergence of $f(z)=\Delta^d(z)$ via Algorithm \ref{alg:radius}.  The algorithm   returns additional information (truncated Puiseux series at each relevant candidate singularity), essential for the asymptotic analysis further below.

{\small
\begin{algorithm}[t]
\caption{Radius of convergence and dominant singularities}
\label{alg:radius}
\begin{algorithmic}[1]
\Statex \textbf{Input:} $K(z,u)$, a nonzero, square-free kernel polynomial;
$Q(z,u_1,\dots,u_c) = T(z,u_1,\dots,u_c)/D(z,u_1,\dots,u_c)$,  a rational expression (both from Alg.~\ref{alg:ana}); $p\geq 0$, a rational (truncation order).
\Statex \textbf{Output:}
A set $S = \{(\zeta_1,S_{\zeta_1}(z)),\dots,(\zeta_k,S_{\zeta_k}(z))\}$ of points $\zeta_i$ of modulus $|R|$,  together with the   Puiseux series of  $\Delta^d(z)$ at those points  truncated at order $p$, s.t.  $\{\zeta_1,...,\zeta_k\}$ includes all dominant singularities of $\Delta^d(z)$;  or {\textbf{Fail}}.
\State  \label{step:ideal} Let
{  $\Lambda[K,D]  :=   \bigl\{ z \in \mathbb{C} \mid \exists\, u_1,\dots,u_c \in \mathbb{C} \text{ s.t. } K(z,u_1)=  \cdots=K(z,u_c)=D(z,u_1,\dots,u_c)=0 \bigr\}$}.
Compute an ideal $J \subseteq \mathbb{C}[z]$ eliminating $u_1,\dots,u_c$ from $\{K(z,u_1),...,K(z,u_c),D(z,u_1,...,u_c)\}$.
{\textbf{If}} $J = \{ 0 \}$ {}\textbf{then return} {\textbf{Fail}}.
\State \label{step:roots} $\Sigma := \Xi[K] \cup \Lambda[K,D]$,  is a \emph{finite} set of singularity candidates of $\Delta^d(z)$. Partition $\Sigma$ into subsets of elements having equal modulus, discarding any subset that does not contain a real element $\ge 1$. Denote the remaining subsets by $\Sigma_1, \dots, \Sigma_t$, ordered   by increasing modulus.
\State   Using the Newton's polygon algorithm, compute  all $\delta_u$ distinct Puiseux series of  $K(z,u)$ at $z=0$, and identify  among them the $c$ \emph{small branches} $u_1,...,u_c$ as those that are nonzero with  a  valuation $>0$.
\For{$i = 1$ \textbf{to} $t$}
\State \label{step:pring}
Let \(\zeta\) be the unique positive real element of \(\Sigma_i\).
\State \label{step:polyg}  Using the Newton's polygon algorithm, compute all $\delta_u$ distinct Puiseux series of  $K(z,u)$ at $z=\zeta$.
\State 
\label{step:AC}
Let $S_i:=\AC(K,\zeta,\gamma,u_i)$ ($i=1,...,c$ and $\gamma$ as per Def. \ref{def:conn}), be the Puiseux branch at \(\zeta\) analytically connected to \(u_i\) along \(\gamma\). 
\State \label{step:series}  Compute
$
S_\zeta(z) := Q\bigl(z, S_1(z), \dots, S_c(z)\bigr) = \sum_{n \ge n_0} c_n (z-\zeta)^{n/\kappa}
$,
working in the field of formal Puiseux series in $(z-\zeta)$. Truncate  $S_\zeta(z)$ to degree $p$ included, using enough terms of the series $S_j(z)$'s from step \ref{step:AC} for output.
\If{$S_\zeta(z)$ is \emph{not} a Taylor series (cf. App. \ref{app:details})}{}\label{step:analytic}
\State   $S:=\{(\zeta,S_\zeta)\}$\label{step:initS}
\For{$\zeta'\in \Sigma_i$, $\zeta'\neq \zeta$}\label{step:other}
\State \label{step:other2} Compute $S_{\zeta'}$ repeating  steps \ref{step:polyg}, \ref{step:AC}, \ref{step:series} (with $\zeta'$ in place of $\zeta$) then update $S:=S\cup \{(\zeta',S_{\zeta'})\}$.
\EndFor
\State\label{step:return}  \textbf{return} $S$
\EndIf
\EndFor
\State \label{step:final}\textbf{return} $\emptyset$ \Comment{no singularity in $\C$, $R = +\infty$}
\end{algorithmic}
\end{algorithm}}
\noindent

\begin{theorem}[correctness and 1-state completeness]\label{th:radius_corr}  (1) Assume Algorithm \ref{alg:radius} does not raise \textbf{Fail}, and returns $S=$ $\{(\zeta_1,S_{\zeta_1}(z)),...,\zeta_k,S_{\zeta_k}(z)\}$. Then $\{\zeta_1,...,\zeta_k\}$ includes  the dominant singularities of $\Delta^d(z)$, and   each $S_{\zeta_i}(z)$ is a  truncated Puiseux series of $\Delta^d(z)$ at $z=\zeta_i$. 
If $S\neq \es$ then $R=|\zeta_i|$ (for any $i$) is the finite radius of convergence of $\Delta^d(z)$. If $S=\es$ then  $R=+\infty$.
(2) Assume $K,Q$ are returned by Alg. \ref{alg:ana} for a nontrivial 1-state program. Then Algorithm \ref{alg:radius}  run with $K,Q$ does not raise \textbf{Fail}.
\end{theorem}

\begin{remark}[computational details, complexity]\label{rem:symbolicvsexact}{
A few computational considerations are in order.
\begin{itemize}
\item
Step \ref{step:ideal}  not raising  {\textbf{Fail}}  guarantees that  $\Lambda[K,D]$ is finite. Here we make use of the  standard notion of \emph{elimination ideal}   from algebraic geometry:
see \cite[Ch.3]{Cox} for a precise definition of ideal and elimination algorithms. In essence, the result  $J$ of the elimination procedure  is a set of   polynomials in $z$, whose common roots coincide with $\Lambda[K,D]$. By the fundamental theorem of algebra, the only case when this set of roots is infinite is when $J$ consists of the zero polynomial.
Although multivariate elimination is computationally hard in general, the present system has a particularly favourable structure. The equations \(K(z,u_i)=0\) are separated in the variables \(u_i\), with coupling occurring only through the single equation \(D(z,u_1,\ldots,u_c)=0\). Hence elimination here can be realized through successive univariate resultants. This is cheap in practice and avoids much of the intermediate expression growth typical of Gr\"{o}bner elimination for an unstructured multivariate system.
\item In step  \ref{step:roots}, one must  handle  roots of polynomials. As standard in symbolic computation, one does not compute such  roots   explicitly,    rather one works in  the algebraic number field   obtained by extending   $\QQ$ with them.
\item In step \ref{step:pring}, at each iteration, we only consider the unique positive real candidate  $\zeta$ in $\Sigma_i$. If $\zeta$ is not a singularity, then $\Sigma_i$ cannot contain  any singularity  (Pringsheim theorem, Lemma \ref{lemma:R}(2,b)), and can be discarded right away. Otherwise, $\zeta=R$, and $\Sigma_i$ is a superset of all dominant singularities.
\item In step \ref{step:AC}, computing  $\AC$ can be reduced to   \emph{Homotopy Continuation (HC)}, see e.g \cite{bates2023numerical,duff2023polynomial}, a problem   well-studied in numerical algebraic geometry. Effective algorithms and software tools for it exist. We provide further details about HC in  Appendix \ref{app:details}.
\item In   step \ref{step:analytic}, deciding if $S_\zeta(z)$ is (not) a Taylor series,  i.e. is (non)analytic, is based on Theorem \ref{th:singAng}(1).
    The test is described in detail in Appendix  \ref{app:details}.
\item In step \ref{step:other}, the (truncated) Puiseux series at all   points in $\Sigma_i$ other than $\zeta$ are computed.
\end{itemize}
\vspace*{.1cm}
\noindent
The   complexity of Algorithm \ref{alg:radius} depends on the   specific data structures   and   algorithms used to represent and manipulate algebraic numbers, and on the adopted Homotopy Continuation method   hinted above. In any case,  the elimination step \ref{step:ideal} alone has an exponential worst-case complexity, although the practical complexity is low as argued above. We leave the details of complexity for future work.}
\end{remark}

When $R>1$, the result  of Algorithm \ref{alg:radius}   gets us     immediately   exponential bounds on the coefficients of $\Delta^d(z)$, based on Lemma \ref{lemma:R} 2(a). Indeed, it is sufficient to select $r\in (1,R)$ to get an exponentially decreasing bound $c_n\leq \Delta^d(r)/r^n$. Note that according to  Theorem \ref{th:singAng}(2), for $r$ inside the disk of convergence, $\Delta^d(r)$ can be computed as  the first nonzero coefficient of the Puiseux expansion at $z=r$ of $f(z)=\Delta^d(z)$.
When $S\neq \es$,  $R$ is finite and, regardless of whether $R=1$ or $R>1$,  we can apply the standard \emph{transfer method} from analytic combinatorics to compute  asymptotic formulae in $n$ for the coefficients $c_n$. We give a quick review of this method. Roughly speaking, one starts from $(\zeta,S_\zeta(z))\in S$ and, after rewriting   $(z-\zeta)^\alpha$ as $(-\zeta)^{-\beta}\cdot(1-z/\zeta)^{-\beta}$ ($\beta= -\alpha$), builds the asymptotic formula in $n$   term-by-term, by applying the following correspondence:
{
\begin{align*}
a\cdot (1-z/\zeta)^{-\beta} & \ \leadsto \ \ a\cdot \zeta^{-n} \cdot \frac{n^{\beta-1}}{\Gamma(\beta)} \ \ \ \text{ (provided $\beta\notin\Z$). }
\end{align*}}\noindent
Here the Gamma function $\Gamma(z)$     is the analytic extension of   factorial   to the complex plane.
We have the following classical result, see  \cite[Th.VI.4,Th.VI.5,Th.VII.8]{Flajo}.

\begin{theorem}[asymptotic analysis: transfer theorems]\label{th:asympt}
Let $f(z)$ be  an algebraic function analytic at $z=0$,  with a \emph{finite} radius of convergence $R>0$ of its Taylor series, $f(z)=\sum_{n\geq 0} c_n z^n$. Let $Z\neq \es$ be the  set of its dominant singularities.
For each   $\zeta\in Z$, consider the Puiseux series at $z=\zeta$ of $f(z)$ truncated at  order $p$, written as:  $S_\zeta(z)=\sum_{-\beta\leq  p}a_{\zeta,\beta}\cdot (1-z/\zeta)^{-\beta}$.
Then, provided the   RHS is not identically zero (as a function of $n$):
\begin{equation}\label{eq:asympt}
c_n \sim \; \sum_{\zeta\in Z,\,-\beta\leq p,\,\beta\notin \Z}   a_{\zeta,\beta}\cdot \zeta^{-n}\cdot \frac{n^{\beta-1}}{\Gamma(\beta)}.\end{equation}
\end{theorem}

The higher the truncation order $p$, the more accurate the approximation; but the asymptotic statement of Theorem    \ref{th:asympt} holds true for \emph{any} order $p$ yielding a nonzero sum. Also note that, in our application of this theorem, given the set $S$ returned by Algorithm \ref{alg:radius} we will let $Z=\{\zeta_1,...,\zeta_k\}$, which  in general is a \emph{superset} of the actual dominant singularities. The presence of non-singularities in $Z$ is harmless: they will not contribute to the sum in \eqref{eq:asympt}, as their truncated Puiseux series only contain terms with nonnegative integer exponents, which are filtered out by the  condition $\beta\notin \Z$.

\begin{example}\label{ex:asymp} For the random walk program from Example \ref{ex:run6}, Algorithm \ref{alg:ana} returns $K(z,u)=  -u^2z + 2u - z$, $c=1$ and $Q(z,u_1)=u_1$, as seen. Algorithm \ref{alg:radius}   returns $S=\{(-1,S_{-1}(z)), (1,S_1(z))\}$, hence $R=1$, with $S_{-1}(z)=-1 + \sqrt{2} \sqrt{z + 1}$ and $S_{1}(z)=1 + {\mathlarger{\iota}} \sqrt{2}\sqrt{z - 1}$. Applying the transfer method (\highlight{Th.} \ref{th:asympt}), we obtain $\Delta^0(z)=\sum_{n\geq 0} c_n\cdot z^n$ with
{
$$c_n \ \ \sim \ \ \frac{ (1-(-1)^{n}) \sqrt{2} }{2\sqrt{\pi} n^{3/2}}$$}\noindent
Additionally, for any $d$ one can evaluate $\Delta^d(1)=S_1(1)=1$ and $\frac {\mathrm{d}}{\mathrm{d} z}\Delta^d(1)=\frac{\mathrm{d}} {\mathrm{d} z} S_1(1)=\infty$.
For the 2-state running example, computation of the exact radius  is highly impractical, as  11th degree polynomials are involved. We can nonetheless consider the smallest real candidate singularity  $R_0\in \Sigma_1$, according to the terminology of Algorithm \ref{alg:radius}, as a lower-bound of the exact radius. We find $R_0=-\frac 4 3+\frac 2 3\sqrt{13}\approx 1.0704$.
 Based on Th. \ref{th:singAng}(2), for any given $d$ we can compute   $\Delta^d(R_0-\epsilon)$ (with $\epsilon \in (0,R_0-1)$; e.g. $\Delta^0(1.05)\approx 0.87991$) and, based on Lemma \ref{lemma:R}(2,a) we can still conclude
{
$$c_n \ \leq \ \ \frac{\Delta^d(R_0-\epsilon)}{(R_0-\epsilon)^n}\ \ = \ \ 0.87991 ...\times 0.95238... ^n.$$}\noindent
\end{example}

\section{Enhancements, tail bounds and further examples}\label{sec:en}
{\sc Maple} code for the examples presented in this section is available online\footnote{\url{https://osf.io/3kfr2/overview?view_only=1e39537d4d44485e98d288922f5286d8}.}.

\paragraph{Simplification}\label{sec:ext}
Algorithm~\ref{alg:ana} fails if there are more $M^d_{j,h}(z)$'s in~\eqref{eq:system} than small
branches. As an example, consider the variation of our running example with transitions from state $q=2$
switched (see below); its kernel polynomial has only $c=1$ small branch, while Theorem~\ref{th:char2}
involves three gf's $M^d_{1,0}(z),M^d_{2,0}(z),M^d_{2,1}(z)$.
\begin{lstlisting}[mathescape=true,basicstyle=\ttfamily\footnotesize]
    while($d\geq0$){
        if ($q=1$)
            {$\frac{1}{2}$:$d:=d+1\,;\,q:=2$[]$\frac{1}{2}$:$d:=d-1\,;\,q:=1$}
        else
            {$\frac{2}3$:$d:=d+1\,;\,q:=2$[]$\frac{1}6$:$d:=d-2\,;\,q:=1$[]$\frac{1}6$:$d:=-1$}}
\end{lstlisting}
Algorithm \ref{alg:ana} can be enhanced with a  {\emph{simplification} method}, that tries to    eliminate  enough $M^d_{j,h}(z)$'s to  make the system solvable by: (1) expressing $M^d_{j,h}(z)$ in
terms of the bivariate $M^d_j(z,u)$; (2) expressing $M^d_j(z,u)$ in terms of $M^d_1(z,u)$ via
system~\eqref{eq:Meq}. The method is sound, but not guaranteed to succeed in general. For the above example, this yields $M^d_{2,0}(z)=0$ and $M^d_{2,1}(z)=\tfrac 1 2 z M^d_{1,0}(z)$,
reducing to a single-equation system solvable with $c=1$ small branch, and with a 2nd degree kernel polynomial which can be solved explicitly.  This leads to explicit formulae for $\Delta^d(z),\Delta'^d(z)$, omitted here. Applying Algorithm \ref{alg:radius}, we get  $R= \frac{7-\sqrt{13}}{3}\approx 1.131$ and
{\begin{equation*}
c_n \ \ \sim \  \frac{ -247+61 \sqrt{13}  }{-3276+468 \sqrt{13}} \cdot R^{-n}.
\end{equation*}}

\paragraph{Language extension}\label{sec:gcp+}
A natural extension is to allow \emph{nonnegative absolute updates}: assignments of the form $d:=k$, for any integer $k\geq 0$. An absolute
update $p:d:=k$ at state $j$ gives a term $p\cdot u^k\cdot M^d_j(z)$ in~\eqref{eq:Meq0}; the
$M^d_j(z)$'s are eliminated via the identity $M^d_j(z)=M^d_j(z,1)$.
As an example of GCP program falling in  this format, consider the program in Fig. \ref{fig:zeroconf} (left), a version of the  {\emph{zeroConf}}   from \cite[Fig.12]{batz2023inductive}. In this      protocol,  a client repeatedly tries to establish  a network connection to a server,  up  to a maximum of $C$ attempts. The counting of attempts can be probabilistically reset. In terms of running time, our program in Fig. \ref{fig:zeroconf} (left), when  executed with initial values $d=C-1$, $q=1$,  is equivalent to the program  in \cite[Fig.12]{batz2023inductive}. 
Note that in our program the \texttt{else} branch involves an absolute update that   `resets' the counter to $C-1$. 
%
Applying Algorithm \ref{alg:ana},   extended as outlined above,
returns a pair $(K,Q)$, with $K$  a 1st degree polynomial in $u$ that can be solved directly for $u(z)$.  This leads to the   following explicit generating function, which happens to be rational; note that we must consider the initial counter value $d=C-1$, as explained above. Algorithm \ref{alg:ana}  can be implemented to treat $\lambda$ and  $C\geq 1$ symbolically:
$$\Delta^{C-1}(z)  = {\frac { \left( \left( 1-\lambda \right) {z}^{2}+2\,\lambda\,z
-2\right)  \left( \lambda\,z \right) ^{C}+{z}^{2} \left(\lambda -1
\right)  }{{z}^{2} \left( \lambda-1  \right)  \left( \lambda\,z
\right) ^{C}+ \left( 1-\lambda \right) {z}^{2}+2\,\lambda\,z\,-2}}.$$
In order to apply Algorithm \ref{alg:radius} and determine $R$, we  must fix specific values for $C$ and $\lambda$. For instance, with $C=10,\lambda=9/10$, we get $R\approx 1.1929$ (the exact value being  a root of an 11-th degree polynomial) and   the asymptotics
{
\begin{equation*}
c_n \ \sim\  \gamma\cdot R^{-n} \ \ \ \ \ \text{($\gamma\approx 0.4285$).}
\end{equation*}}\noindent
%
\usetikzlibrary{automata, positioning}
\begin{minipage}{\columnwidth}
\begin{minipage}[c]{0.5\columnwidth}
\footnotesize
\begin{lstlisting}[
mathescape=true,
basicstyle=\footnotesize\ttfamily,
columns=fullflexible,
%caption={Program $P$},
label={lst:template}
]
while ($d\geq 0$) {
    if ($q=1$)
        {$\lambda$: $d:=d-1$ [] $1-\lambda$: $d:=C-1$; $q:=2$}
    else
        {1/2: $q:=1$ [] 1/2: $d:=-1$} }
\end{lstlisting}
\end{minipage}
\begin{minipage}[c]{0.5\columnwidth}

\begin{center}
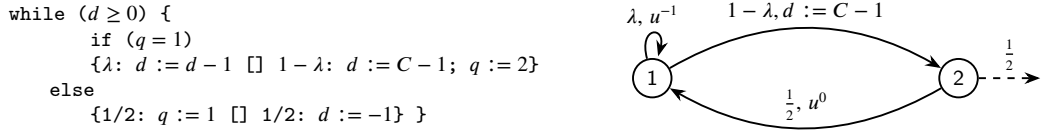

\scalebox{.9}{
\begin{tikzpicture}[->, >=Stealth, auto, node distance=4.5cm, semithick,every node/.style={scale=1.0, text=black}, thick, draw=black]
	\tikzstyle{every state}=[fill=white, draw=black, text=black, minimum size=1pt, draw=black]

\node[state] (B) {1};
\node[state, right of=B] (A) {2};

\path (A) edge[bend left=30] node[above] {$\frac{1}{2},\,u^{0}$} (B);

\path (B) edge[bend left=30] node[above] {$1-\lambda, d:=C-1$} (A);

\path (B) edge[loop above] node[align=center] {$\lambda,\,u^{-1}$} (B);

\path (A) edge[dashed] node[above] {$\frac{1}{2}$} ++(1.2,0);
\end{tikzpicture}}
\end{center}
\end{minipage}
\captionof{figure}{\small\textbf{Left}: zeroConf protocol, adapted from \cite{batz2023inductive}; $\lambda\in (0,1)$ and   integer $C\geq 1$ are   parameters. \textbf{Right}:    \highlight{\poc} representation; the label  $d:=C-1$ marks the transition arising from an absolute update.}
\label{fig:zeroconf}
\end{minipage}

\paragraph{Comparison of tail bounds} The only study we are aware of concerning the distribution of running time in pPDAs and related formalisms is  Br\'{a}zdil et al.'s \cite{BrazdilKieferKuceraVarekova15}. It is interesting to compare our tail bounds,  as implied by Lemma \ref{lemma:R}(2,a), with theirs. For an almost surely terminating 1-state pPDA -- pBPA, in their terminology -- the running-time random variable $\mathrm{RT}$  satisfies, for sufficiently large $n$ (see \cite[Th.4.1]{BrazdilKieferKuceraVarekova15}):
{
\begin{equation}\label{eq:esparza}
\Pr(\mathrm{RT} \geq n) \;\leq\; \mathrm{e}\cdot \mathrm{e}^{- \frac{n}{8E_{\max}^2}}
\end{equation}}\noindent
where $E_{\max}$ is the maximal expected running time, taken over all possible starting stack symbols. This bound  is generic: it will yield  the same exponential decrease rate for (even wildly) different pBPAs with the same $E_{\max}$. In particular, as for nontrivial pBPAs $E_{\max}\geq 1$, the RHS of \eqref{eq:esparza}  cannot decrease faster than $\sim \mathrm{e}^{-\frac n 8}\approx 0.8824^n$. On the other hand, our asymptotic analysis delivers tight results.

A simple example is sufficient to appreciate this difference. Consider
the   pBPA with one stack symbol $X$  defined by the two rules:
$X \to \epsilon$  (prob. $1/3$), $X\to  X$  (prob. $2/3$). Termination of a pBPA occurs   as soon as the empty string $\epsilon$ is produced. So this pBPA is equivalent,  in terms of running time, to the following GCP, with initial $d=0$:
\begin{lstlisting}[mathescape=true,basicstyle=\ttfamily\footnotesize]
while $(d\geq 0)$ { 1/3: $d:=d-1$ [] 2/3: $d:=d$ }
\end{lstlisting}
Now $E_{\max}=3$, hence \eqref{eq:esparza} yields the bound $a_n:=\mathrm{e}\cdot \mathrm{e}^{-\frac n{72}} = \mathrm{e}\cdot 0.9862^n$.
Application of our Algorithms  \ref{alg:ana} and then \ref{alg:radius} yields $R=3/2$. Choose e.g. $R_0:=1.4$ and compute,  via Theorem \ref{th:singAng}(2), $\Delta^0(R_0)=7$. Then Lemma \ref{lemma:R}(2,a) guarantees  $c_n\leq 7/R_0^n$ for the coefficients of $\Delta^0(z)$. Standard algebraic manipulations yield the $\mathrm{RT}$-tail bound: $\Pr(\mathrm{RT}>n)\leq b_n:=\,\frac{7}{R_0^n(R_0-1)}$. The ratio $b_n/a_n$ is $<1$ already for $n\geq 6$, the smallest $n$ for which \eqref{eq:esparza} is guaranteed. And decreases exponentially; e.g.  $b_{30}/a_{30}<0.0004$. The transfer theorem gives the exact asymptotics
{%
$$
c_n\ \ \sim \ \ \frac12 \left(\frac 2 3\right)^n
\,.$$}
\section{Conclusion}\label{sec:concl}
We have presented an algebraic framework for exact quantitative analysis of GCP programs, grounded in an
operational semantics via \highlight{\poc} runs. The key insight is that the operational generating function of a GCP
program is algebraic, characterised via a kernel polynomial whose analysis yields exact asymptotic information  and exponential tail bounds. The method is sound for GCP programs and complete
for 1-state programs.

Future work includes extending the framework to broader classes of probabilistic loops, developing
automated detection of suitable program structures (cf.\ the zeroConf example), and exploring
hybrid approaches combining our exact analysis with martingale-based or proof-system techniques. An interesting possibility is translating probabilistic programs into polynomial systems, then analysing directly the singularities of these systems. Some preliminary work in this direction is reported in \cite{BCP26}. A major issue in this approach is how to avoid  the cost of variable elimination. In this respect, connections with work in analytic combinatorics \cite{BanderierDrmota2015,PivoteauSalvy2025} deserve  further attention.

\vsp
\subsection*{Dedication}
This paper is dedicated to Gianluigi Ferrari on the occasion of his 65th birthday.
The first author (M.B.) fondly recalls discussions with Giangi over the years,
on topics ranging from bisimilarity in the pi-calculus to the best
places to eat testaroli.

\bibliographystyle{unsrt}
\bibliography{references}

\newpage
\appendix
\section{Proofs}\label{app:proofs}
\subsection{Section \ref{sec:GCPP}}\label{app:sec:GCPP}

\begin{proof_of}{Theorem \ref{th:char2}}
For brevity, write
\[
D_n^d := D_{1,n}^d .
\]
If \(d<0\), then by definition \(D_n^d=[n=0]\), hence
\[
\Delta^d(z)=\sum_{n\ge 0} D_n^d z^n = 1.
\]
Assume now \(d\ge 0\). We claim that, for every \(n\ge 0\),
\begin{equation}\label{eq:coeff-rec-thm5}
D_{n+1}^d
=
\sum_{j=1}^N
\left(
\sum_{h=0}^{e_j-1} p_j^{<-h}\, m_{j,n,h}^d
\;+\;
b_j\, m_{j,n}^d
\right).
\end{equation}
Indeed, let \(\tau\) be a terminal run of length \(n+1\), starting from
state \(1\) and altitude \(d\). Since \(d\ge 0\), \(\tau\) has a unique
decomposition
\[
\tau=\sigma,\eta,
\]
where \(\sigma\) is the prefix of length \(n\), and \(\eta\) is the last
transition. Because \(\tau\) is terminal and all proper prefixes of a
terminal run are non-terminal, \(\sigma\) is a meander. Say \(\sigma\)
ends in state \(j\) and altitude \(h\ge 0\). Then \(\eta\) is of exactly
one of the following two kinds.
\noindent
\emph{(i) Absolute terminating update.}
The last transition is an absolute update from state \(j\), which occurs
with total probability \(b_j\). The total weight of meanders \(\sigma\)
of length \(n\) ending in state \(j\) and arbitrary nonnegative altitude
is \(m_{j,n}^d\). Therefore these runs contribute
\[
b_j\, m_{j,n}^d .
\]
\noindent
\emph{(ii) Relative terminating update.}
The last transition is a relative update \(k<0\) from state \(j\) such that
\(h+k<0\), equivalently \(k<-h\). By definition, the total probability
of all such relative updates from state \(j\) is \(p_j^{<-h}\). The total
weight of meanders \(\sigma\) of length \(n\) ending in state \(j\) and
altitude \(h\) is \(m_{j,n,h}^d\). Hence these runs contribute
\[
p_j^{<-h}\, m_{j,n,h}^d .
\]
Moreover, if \(h\ge e_j\), no such transition exists: indeed, by
definition of \(e_j\), every negative relative update \(k\) from state \(j\)
satisfies \(k\ge -e_j\), hence \(h+k\ge 0\). Therefore only
\(h=0,\dots,e_j-1\) can contribute.
The two cases are disjoint and exhaustive, and summing over all
possible final states \(j\) and admissible altitudes \(h\) yields
\eqref{eq:coeff-rec-thm5}.
Now multiply \eqref{eq:coeff-rec-thm5} by \(z^{n+1}\) and sum over all
\(n\ge 0\). Using the definitions of \(\Delta^d(z)\), \(M_{j,h}^d(z)\) and
\(M_j^d(z)\), we obtain
\[
\sum_{n\ge 0} D_{n+1}^d z^{n+1}
=
z\sum_{j=1}^N
\left(
\sum_{h=0}^{e_j-1} p_j^{<-h} \sum_{n\ge 0} m_{j,n,h}^d z^n
\;+\;
b_j \sum_{n\ge 0} m_{j,n}^d z^n
\right),
\]
that is,
\[
\Delta^d(z)-D_0^d
=
z\sum_{j=1}^N
\left(
\sum_{h=0}^{e_j-1} p_j^{<-h} M_{j,h}^d(z)
\;+\;
b_j M_j^d(z)
\right).
\]
Since \(d\ge 0\), we have \(D_0^d=0\). Therefore
\[
\Delta^d(z)
=
z\sum_{j=1}^N
\left(
\sum_{h=0}^{e_j-1} p_j^{<-h} M_{j,h}^d(z)
\;+\;
b_j M_j^d(z)
\right).
\]
Combining this with the already established case \(d<0\) gives
\[
\Delta^d(z)
=
[d<0]\cdot 1
+
[d\ge 0]\cdot
z\sum_{j=1}^N
\left(
\sum_{h=0}^{e_j-1} p_j^{<-h} M_{j,h}^d(z)
\;+\;
b_j M_j^d(z)
\right),
\]
as claimed.
\end{proof_of}

\subsection{Section \ref{sec:alg}}\label{app:sec:alg}
\begin{proof_of}{Lemma \ref{prop:Meqs}}
Fix \(i\in\{1,\dots,N\}\). We first derive a recurrence for the coefficients
\(m_{i,n,h}^d\). Consider a meander \(\sigma\) counted by \(m_{i,n,h}^d\),
that is,
\[
|\sigma|=n,\qquad \mathrm{ia}(\sigma)=d,\qquad \mathrm{is}(\sigma)=1,\qquad \mathrm{fa}(\sigma)=h,\qquad \mathrm{fs}(\sigma)=i.
\]
For \(n=0\), the only meander of length \(0\) starts in state \(1\), stays at
altitude \(d\), and ends in state \(1\). Hence
\[
m_{i,0,h}^d=[i=1][h=d].
\]
Let now \(n\ge 1\)  and consider a meander \(\sigma\) counted by \(m_{i,n,h}^d\).
Write \(\sigma=\sigma',\eta\), where \(\eta\) is the last transition of \(\sigma\).
Then \(\sigma'\) is a meander of length \(n-1\), ending in some state \(j\) and
some altitude \(h-k\ge 0\), and \(\eta\) is a transition
\[
(p,k,i)\in S_j
\]
of weight \(p\), sending state \(j\) to state \(i\) and changing the altitude by
\(k\). Conversely, every such pair \((\sigma',\eta)\) yields a unique meander
ending in state \(i\) at altitude \(h\). Therefore
\[
m_{i,n,h}^d
=
\sum_{j=1}^N \ \sum_{(p,k,i)\in S_j}
p\,[h-k\ge 0]\; m_{j,n-1,h-k}^d ,
\qquad (n\ge 1,\ h\ge 0).
\]
Combining this with the case \(n=0\), we obtain for all \(n,h\ge 0\):
\begin{equation}\label{eq:coeff-lemma6}
m_{i,n,h}^d
=
[i=1][n=0][h=d]
+
[n\ge 1]\sum_{j=1}^N \ \sum_{(p,k,i)\in S_j}
p\,[h-k\ge 0]\; m_{j,n-1,h-k}^d .
\end{equation}
We now multiply \eqref{eq:coeff-lemma6} by \(z^n u^h\) and sum over all
\(n,h\ge 0\). On the RHS, the first term contributes
\[
\sum_{n,h\ge 0} [i=1][n=0][h=d]\, z^n u^h = [i=1]u^d .
\]
For the second term in \eqref{eq:coeff-lemma6}, sum rearranging yields
\begin{align}\label{eq:meand1}
\sum_{n,h\ge 0}
[n\ge 1]
\sum_{j=1}^N \sum_{(p,k,i)\in S_j}
p\,[h-k\ge 0]\; m_{j,n-1,h-k}^d \, z^n u^h\;=
\\
\hspace*{2cm}\nonumber
z\sum_{j=1}^N \sum_{(p,k,i)\in S_j}
p\sum_{n\ge 1}\sum_{h\ge 0}
[h-k\ge 0]\; m_{j,n-1,h-k}^d \, z^{n-1} u^h .
\end{align}
Now fix \(j\in\{1,...,N\}\) and \((p,k,i)\in S_j\), and set \(m=n-1\), \(r=h-k\). Since
\(n\ge 1\) and \(h-k\ge 0\), this gives \(m\ge 0\) and \(r\ge 0\), while
\(h=r+k\ge 0\) forces \(r\ge -k\). Hence \(r\ge \max(0,-k)\), and therefore, for the two inner sums in \eqref{eq:meand1} we have, interchanging the sum in the second step, and using the definition of $M^d_{j,r}$:
\begin{align*}
\sum_{n\ge 1}\sum_{h\ge 0}
[h-k\ge 0]\; m_{j,n-1,h-k}^d \, z^{n-1} u^h
& =
\sum_{m\ge 0}\sum_{r\ge \max(0,-k)}
m_{j,m,r}^d\, z^m u^{r+k}\\
&=\sum_{r\ge \max(0,-k)} M_{j,r}^d(z)\, u^{r+k}.
\end{align*}
Thus continuing from \eqref{eq:meand1}
\[
\sum_{n,h\ge 0}
[n\ge 1]
\sum_{j=1}^N \sum_{(p,k,i)\in S_j}
p\,[h-k\ge 0]\; m_{j,n-1,h-k}^d \, z^n u^h\,
=\,
z\sum_{j=1}^N \sum_{(p,k,i)\in S_j}
p \sum_{r\ge \max(0,-k)} M_{j,r}^d(z)\, u^{r+k}.
\]
Hence
\begin{equation}\label{eq:splitk}
M_i^d(z,u)
=
[i=1]u^d
+
z\sum_{j=1}^N \sum_{(p,k,i)\in S_j}
p \sum_{r\ge \max(0,-k)} M_{j,r}^d(z)\, u^{r+k}.
\end{equation}
We now split according to the sign of \(k\).
If \(k\ge 0\), then \(\max(0,-k)=0\), so
\[
\sum_{r\ge \max(0,-k)} M_{j,r}^d(z)\, u^{r+k}
=
u^k \sum_{r\ge 0} M_{j,r}^d(z)\,u^r
=
u^k M_j^d(z,u).
\]
If \(k<0\), then \(\max(0,-k)=-k\), so
\[
\sum_{r\ge \max(0,-k)} M_{j,r}^d(z)\, u^{r+k}
=
\sum_{r\ge -k} M_{j,r}^d(z)\, u^{r+k}
=
u^k M_j^d(z,u)-\sum_{r=0}^{-k-1} u^{k+r} M_{j,r}^d(z).
\]
Substituting these two identities into \eqref{eq:splitk}, we obtain
\[
M_i^d(z,u)
\;=\;
[i=1]u^d
+
z\sum_{j=1}^N \sum_{(p,k,i)\in S_j} p\,u^k\, M_j^d(z,u)  \
-\
z\sum_{j=1}^N \sum_{\substack{(p,k,i)\in S_j\\ k<0}}
p\sum_{r=0}^{-k-1} u^{k+r} M_{j,r}^d(z).
\]
By the definition of
$S_{ij}(u) =\sum_{(p,k,i)\in S_j} p\,u^k$,
the middle term on the RHS above is
$
z\sum_{j=1}^N S_{ij}(u)\, M_j^d(z,u)
$,
while the last term is exactly \(B_i(z,u)\). Therefore we get
\[
M_i^d(z,u)
=
[i=1]u^d
+
z\sum_{j=1}^N S_{ij}(u)\, M_j^d(z,u)
-
B_i(z,u),
\]
as claimed.
\end{proof_of}

\vsp
\vsp
\vsp
\vsp
\begin{proof_of}{Proposition \ref{prop:rational}}
Let
$
m:=|\widetilde M^d_{j,h}(z)|
$
be the number of unknowns in~\eqref{eq:system}. Since the rank of~\eqref{eq:system} is at least
$m$,  we may select $m$ equations whose
coefficient matrix is nonsingular over $\C\{\!\{z\}\!\}$. Assume without loss of generality  that $m=c$, so that the selected equations coincide with system~\eqref{eq:system}.
By construction, the entries of $A(u)$ belong to
$\C[u,u^{-1}]$. Hence
$
\adj(I-zA(u))
$
also has entries in $\C[z,u,u^{-1}]$. Therefore each   equation
of~\eqref{eq:system}, before substituting $u=u_s(z)$, is linear in the unknowns
$\widetilde M^d_{j,h}(z)$ and has coefficients that are Laurent
polynomials in the corresponding indeterminate $u_s$ (and polynomial
in $z$).
Multiply each equation in~\eqref{eq:system} by a sufficiently large power of $u_s(z)$, say
$u^{r_s}_s(s)$, so as to clear all negative powers. This does not change  the set of solutions of ~\eqref{eq:system}, as each small branch
$u_s(z)$ is nonzero, hence invertible in $\C\{\!\{z\}\!\}$. We therefore obtain a new linear system
\begin{equation}\label{eq:system2}
u^{r_s}_s\cdot L(z,u_s(z),\widetilde M_{j,h}(z))=0\quad \quad \text{($s=1,...,c$)}
\end{equation}
with the same solutions as \eqref{eq:system}.
Fixing   indeterminates
$u_1,\ldots,u_c$, consider the system obtained by replacing   $u_s(z)$ with $u_s$ in each equation of \eqref{eq:system2}. We    obtain a   linear system for $\widetilde M^d_{j,h}(z)$  with polynomial coefficients, that is
\begin{equation}\label{eq:systemPoly}
        H(z,u_1,\ldots,u_c)\,X
        =
        a(z,u_1,\ldots,u_c),
\end{equation}
where $H$ is a $c\times c$ matrix, $a$ is a $c$-vector, and all
entries belong to $\C[z,u_1,\ldots,u_c]$.
Consider the ring homomorphism defined by
\[
\sigma:\C[z,u_1,\ldots,u_c]\longrightarrow\C\{\!\{z\}\!\},
\qquad
\sigma(u_s)=u_s(z),\qquad \sigma(z)=z .
\]
With obvious notation, $\sigma(H)X=\sigma(a)$ is system~\eqref{eq:system2}, which as said is equivalent to~\eqref{eq:system}.
Hence
$
        \det \sigma(H)\neq 0
$.
Since
$
        \det \sigma(H)=\sigma(\det H)
$,
we deduce $\det H\neq 0$, i.e. $H$ is invertible over
$$
        F:=\C(z,u_1,\ldots,u_c),
$$
and the polynomial system \eqref{eq:systemPoly} has a unique solution
\[
        \eta:=H^{-1}a\in F^m .
\]
Each component of $\eta$ is a rational expression with $\det H\neq 0$ as a denominator.  As
$
        \sigma(\det H)  \neq 0
$,
  the components of $\eta$ can be specialized at
$u_s=u_s(z)$, i.e.  $\sigma(\eta)$ is well defined, and  is the unique solution of~\eqref{eq:system2} hence  of~\eqref{eq:system}. Hence
\[
        \sigma(\eta)=\eta(z,u_1(z),\ldots,u_c(z))
        =
        \widetilde M^d_{j,h}(z).
\]
Now consider system~\eqref{eq:Meq}. We     substitute the components of $\eta$ for the
$\widetilde M^d_{j,h}(z)$ in~\eqref{eq:Meq}. Since
$
        \det(I-zA(u))
$
is a nonzero element of $\C(z,u)$ (indeed it equals $1$ at $z=0$),
system~\eqref{eq:Meq0} can be solved over the rational function field
$\C(z,u,u_1,...,u_c)$, yielding a   vector $M^d(z,u)$ of rational expressions,  with $\det(I-zA(u))$ at the denominator.
Setting $u=1$ (note that $\det(I-zA(1))\neq 0$ as a polynomial) gives rational expressions for the $M^d_j(z)$'s in
$F$. Finally, substituting these expressions into the formula of
Theorem~\ref{th:char2} yields an element
$
        Q(z,u_1,\ldots,u_c)\in F
$
such that
\[
        \Delta^d(z)
        =
        Q(z,u_1(z),\ldots,u_c(z))
\]
near $z=0$.
\end{proof_of}
\vsp
\begin{lemma}\label{lemma:det} Consider a field that extends $\reals$.
Let $\gamma_1,...,\gamma_k$ be nonzero elements of the field,  and $p_1,...,p_k\in \reals$.  Consider the $k\times k$ matrix $H$ whose $s$-th row  is
$$
r_s\; =\; \Big[\; \sum_{h=1}^k p_h \gamma_s^{-h},\;
\sum_{h=2}^k p_h \gamma_s^{-h+1},\;
\ldots,\;
p_k \gamma_s^{-1} \;\Big]\ \ \ \text{($1\leq s\leq k$).}
$$
Then
\begin{equation}\label{eq:detH}
\det H
= \pm \,
p_k^{\,k}\,
\Bigl(\prod_{s=1}^k \gamma_s^{-k}\Bigr)
\prod_{1\le i<j\le k}(\gamma_i - \gamma_j).
\end{equation}
\end{lemma}

\begin{proof}
Set $\alpha_s:=\gamma_s^{-1}$. Define the matrices
\[
A =
\begin{bmatrix}
\alpha_1^1 & \alpha_1^2 & \cdots & \alpha_1^k\\[2pt]
\alpha_2^1 & \alpha_2^2 & \cdots & \alpha_2^k\\[2pt]
\vdots&\vdots& &\vdots\\[2pt]
\alpha_k^1 & \alpha_k^2 & \cdots & \alpha_k^k
\end{bmatrix}
\ \ \ \text{ and } \ \ \
T =
\begin{bmatrix}
p_1     & p_2     & p_3     & \cdots & p_{k-1} & p_k \\[2pt]
p_2     & p_3     & p_4     & \cdots & p_k     & 0   \\[2pt]
\vdots  & \vdots  & \vdots  & \ddots & \vdots  & \vdots \\[2pt]
p_{k-1} & p_k     & 0       & \cdots & 0       & 0   \\[2pt]
p_k     & 0       & 0       & \cdots & 0       & 0
\end{bmatrix}.
\]
Then
\[
H = A\,\cdot\,T
\]
which implies
\begin{align}\label{eq:detH1}
\det H = \det A \cdot \det T \,.
\end{align}
The matrix \(A\) can be written as
\[
A = \operatorname{diag}(\alpha_1,\dots,\alpha_k)\cdot
\begin{bmatrix}
1 & \alpha_1 & \cdots & \alpha_1^{k-1}\\
1 & \alpha_2 & \cdots & \alpha_2^{k-1}\\
\vdots & \vdots & & \vdots\\
1 & \alpha_k & \cdots & \alpha_k^{k-1}
\end{bmatrix}
= \operatorname{diag}(\alpha_1,\dots,\alpha_k)\cdot\,V
\]
where \(V\) is the usual Vandermonde matrix. It is known that $\det V=\prod_{1\le i<j\le k} (\alpha_j - \alpha_i)$,
hence
\begin{align}\label{eq:detH2}
\det A
= \Bigl(\prod_{s=1}^k \alpha_s\Bigr)
\prod_{1\le i<j\le k} (\alpha_j - \alpha_i).
\end{align}
Replace $\alpha_s=\gamma_s^{-1}$ and note that
$
\alpha_j-\alpha_i = \gamma_j^{-1}-\gamma_i^{-1}=\frac{\gamma_i-\gamma_j}{\gamma_i \gamma_j}
$,
so
$$
\prod_{1\le i<j\le k}(\alpha_j-\alpha_i)
= \frac{\prod_{1\le i<j\le k}(\gamma_i-\gamma_j)}{\prod_{1\le i<j\le k}\gamma_i \gamma_j}
= \frac{\prod_{1\le i<j\le k}(\gamma_i-\gamma_j)}{\prod_{i=1}^k \gamma_i^{\,k-1}}.
$$
Replacing this in \eqref{eq:detH2} and recalling that $\alpha_s=\gamma_s^{-1}$
\[
\det A
= \Bigl(\prod_{s=1}^k \gamma_s^{-k}\Bigr)
\prod_{1\le i<j\le k} (\gamma_i - \gamma_j).
\]
On the other hand, $T$ is an upper anti-triangular matrix, whose determinant is the product of the elements on its main anti-diagonal, up to a sign: $\det T  =  \pm\,p_k^{\,k}$.
This can be shown by e.g. noting that, for a suitable permutation matrix $E$,  $T\cdot E$ is a triangular matrix with $p_k$ on the main diagonal, hence  $p_k^k=\det (T\cdot E) = \det T\cdot \det E$; on the other hand, for any permutation matrix $E$, $\det E=\pm 1$, hence the claimed statement on $\det T$ (an easy induction argument yields the more precise $\det T  = (-1)^{\frac{k(k-1)}{2}}\,p_k^{\,k}$). From \eqref{eq:detH1} and \eqref{eq:detH2} we therefore obtain the wanted \eqref{eq:detH}.
\end{proof}

\begin{lemma}\label{lemma:onestatec}
If $P$ is a nontrivial 1-state GCP program, Algorithm \ref{alg:ana} does not raise \textbf{Fail}.
\end{lemma}
\begin{proof}
Recall that $e$ is the magnitude (absolute value) of the minimal negative update, which we assume $>0$.
Moreover, we have $N=1$. Hence $\det(I-zA(u))=1-zS_{11}(u)$ and $S_{11}(u)=p_eu^{-e}+p_{e-1}u^{-e+1}+\cdots+p_{e-\ell} u^{-e+\ell}$, with $e>0$,  $\ell\geq 0$, $p_i\in[0,1]$, $p_e>0$ and $\sum_j p_j\leq 1$. Then $K(z,u)=u^e+zP(u)$, with $P(u)=p_e+p_{e-1}u+\cdots+p_{e-\ell} u^\ell$. The polynomial $K(z,u)$ has exactly $c=e$ \emph{distinct} small branches: for a proof, see \cite[Prop.4.4]{AsinowskiBacherBanderierGittenberger20}; note that no branch  of $K(z,u)$ can be identically zero, as $K(z,0)=p_e>0$.

Now consider \eqref{eq:Bi}. Consider the magnitudes of the possible \emph{negative} updates, starting from the largest possible: $e,e-1,...,1$; and the corresponding probabilities: $p_e,p_{e-1},...,p_1$.   Rearranging the summands in \eqref{eq:Bi}, we can write $B(z,u):=B_1(z,u)$ as follows, where on the RHS the first row gathers the balancing terms for a relative update $-e$, the second row the balancing terms for a relative update $-e+1$, and so on:
{
$$
\begin{array}{rcllllllllll}
B(z,&u)\\
=&z \Big(&\ \ \ \  \ p_e u^{\text{-}e}M^d_{1,0}(z)&+& \ \ p_e u^{\text{-}e+1}M^d_{1,1}(z)&+&\cdots &+&p_e u^{\text{-}2}M^d_{1,e\text{-}2}(z) &+&p_e u^{\text{-}1}M^d_{1,e\text{-}1}(z)+\\
&     & p_{e\text{-}1} u^{\text{-}e+1}M^d_{1,0}(z)&+& p_{e\text{-}1} u^{\text{-}e+2}M^d_{1,1}(z)&+&\cdots &+&p_{e\text{-}1} u^{\text{-}1}M^d_{1,e\text{-}2}(z)&+\\[2pt]
&&&&&{\;\text{$\vdots$}}\\[2pt]
&      & \ \ \ \ \ p_2 u^{\text{-}2}M^d_{1,0}(z)& + &\ \ \ \ \ p_2 u^{\text{-}1}M^d_{1,1}(z)&+\\[2pt]
&      & \ \ \ \ \ p_1 u^{\text{-}1}M^d_{1,0}(z)\;\Big)\\[4pt]
=& z\,.\, \Big[& \sum_{h=1}^e p_h u^{-h} &, & \sum_{h=2}^e p_h u^{-h+1} &, &  && {\small\ldots} &,& p_e u^{-1} \Big]&\hspace*{-1.3cm}\cdot \, \widetilde{M}^d_{1,h}\\[-3pt]
&&\multicolumn{10}{c}{\hspace*{-1.6cm}\underbrace{\hspace*{11.5cm}}_{\textstyle\ \ \ \ \ \ \  =:r(u)}}
\end{array}
$$}\noindent
where $\widetilde{M}^d_{1,h}(z)=(M^d_{1,0}(z),M^d_{1,1}(z),...,M^d_{1,e-1}(z))^T$, a $e\times 1$  column vector. Now, considering that $w= \mathrm{adj}(I-zA(u))=1$ in this case (as $N=1$), (\ref{eq:MainEq},\ref{eq:system})
give  $L(z,u,\widetilde{M}^d_{1,h}(z))=u^d-B(z,u)=u^d-z\cdot r(u)\cdot \widetilde{M}^d_{1,h}(z)$. Consequently, system \eqref{eq:system} can be written as
\begin{align*}
\mathrm{diag}(z,...,z)\cdot H\cdot \widetilde{M}^d_{1,h}(z) &= v^d
\end{align*}
where $H$ is the $e\times  e$ matrix whose $s$-th row is $r(u_s(z))$ ($1\leq s \leq e$), and $v^d=[u^d_1(z),...,u^d_e(z)]^T$. Note that $H$ is precisely of the form  described in the statement of Lemma \ref{lemma:det}, with $k=e$ and  $\gamma_s=u_s(z)$. Then \eqref{eq:detH}  implies that
$\det H
= \pm \, p_e^{\,e}\,
\Bigl(\prod_{s=1}^e u_s(z)^{-e}\Bigr)\cdot
\prod_{1\le i<j\le e}(u_i(z) - u_j(z))\neq 0$ as an element of $\A$, \replace{since} $p_e>0$ and the $u_s(z)$'s are all distinct and nonzero. This in turn implies that $\det(\mathrm{diag}(z,...,z)\cdot H)=z^e\cdot \det H\neq 0 $. Hence  system  \eqref{eq:system} has rank $e=c=|\widetilde{M}^d_{1,h}(z)|$.
\end{proof}
\vsp

\begin{proof_of}{Proposition \ref{prop:compl}} 
Suppose the algorithm does not raise \textbf{Fail} at step 2. Then Proposition \ref{prop:rational} and Theorem \ref{th:char2} imply immediately the thesis for $Q$, while $K$ is by definition the kernel polynomial. For nontrivial 1-state programs the result follows from Lemma \ref{lemma:onestatec}.
\end{proof_of}

\subsection{Section \ref{sec:algan}}\label{app:sec:algan}
We recall the \emph{identity theorem} for analytic functions: if $f$ and $g$ are defined and analytic in a connected   open subset $U\subseteq \C$, and coincide in a subset $W\subseteq U$ that has an accumulation point (e.g. $W$ is in turn nonempty open), then $f=g$ on $U$;   see e.g. \cite[Ch.10,p.228]{RudinRCA2}. We shall also appeal to the  \emph{monodromy theorem}: suppose $f$ is  analytic in a disk $D$ contained in an open simply connected\footnote{A curve is a continuous function $[0,1]\to \Omega$; simply connected means connected and, informally speaking, without holes; see \cite[p.239]{RudinRCA2} for the formal definition.} set $\Omega$, and can be analytically continued along every curve starting from the center of $D$ and contained in $\Omega$. Then there is a (unique) analytic function $g$ defined in $\Omega$ such that $f=g$ on $D$  \cite[Th.16.15]{RudinRCA2}. For algebraic functions, existence of analytic continuation along any curve is implied by $\Omega\cap \Xi[H]=\es$, where $H$ is a square-free annihilating polynomial for $f$.

\vsp
\begin{proof_of}{Lemma \ref{lemma:AC}} Denote by $v_1,...,v_{\delta_u}$ ($\delta_u=\deg_u(H)$) the branches of $H$ at $z=\zeta$, necessarily distinct   since $H$ is square-free.
As functions, these branches  are defined and analytic in some (slit) disk $D^*$ centered at $\zeta$. As $\gamma(t)\to \zeta$ as $t\to 1^-$, and by the constraints on $\gamma$, there is $0\leq t^*<1$ such that $\tau:=\gamma((t^*,1))\subseteq D^*$.
Take $z^*\in \tau$.   Since $z^*\notin \Xi[H]$, there is exactly one index $j\in\{1,...,\delta_u\}$ such that $u(z^*)=v_j(z^*)$, as $H(z^*,u)$ has $\delta_u$ distinct roots as a polynomial in $u$. By IFT there is a unique analytic branch of \(H\) through \((z^*,u(z^*))\). Since both \(u\) and \(v_j\) pass through that point, they coincide in a neighbourhood of \(z^*\), hence on some $\tau_0\subseteq \tau$.
Consider now an open, simply connected set $W$ s.t. $\tau \subseteq W\subseteq (D_\zeta\cap D^*)\setminus \Xi[H]$; such a $W$ exists because the path $\gamma$ is not self-intersecting. $u$ can   be analytically continued to this $W$  (monodromy theorem). By the identity theorem, since $u$ and $v_j$ coincide on $\tau_0\subseteq W$, they coincide on $W$, hence on  $\tau$. So $v_j$ is analytically connected to $u$ along $\gamma$. Moreover,  $u(\gamma(t))\neq v_i(\gamma(t))$ for all $t\in (t^*,1)$ and $i\neq j$: indeed, if there were $i\neq j$ and $z^{**}=\gamma(t^{**})$ with $t^{**}\in (t^*,1)$ such that $u(z^{**})=v_{i}( z^{**})$,   it would follow that $v_j(z^{**})=v_{i}( z^{**})$, which is impossible, because distinct branches cannot collide at a non-exceptional point $z^{**}$: this would imply that $H(z^{**},u)$ has fewer than $\delta_u$ distinct roots, which is impossible since \(z^{**}\notin\Xi[H]\). So no $v_i$ with $i\neq j$ is analytically connected to $u$ along $\gamma$. This proves the first part of the statement.

Concerning the last part, first note that $u(z)$ is defined and analytic in the whole disk $D_r:=\{z:|z|<r\}$. Consider two paths $\gamma_1$ and $\gamma_2$ taken as per Definition \ref{def:conn}, and let $v_i:=\AC(H,\zeta, \gamma_i,u)$, for $i=1,2$. Using the symbology of the first part, we have $v_1=u$ on $\tau_1$, hence on $D^*\cap D_r$ (identity theorem). For the same reason, $v_2=u$ on $D^*\cap D_r$. So $v_1=v_2$ on $D^*\cap D_r$, hence $v_1=v_2$ on $D^*$ (again identity theorem). Therefore $v_1=v_2$.
\end{proof_of}

\vsp

\begin{proof_of}{Theorem   \ref{th:singAng}} Let $z_0$ and  $\gamma$ be chosen as per Definition \ref{def:conn}. We can also assume that $\gamma$ avoids the sets of points $Z_D$ where the denominator $D(z,u_1(z),...,u_c(z))$ of $Q(z,u_1(z),...,u_c(z))$ vanishes. Note that this denominator cannot be identically zero near $z=0$, because $f(z)= Q(z,u_1(z),...,u_c(z))$ is well defined  there. As a consequence, $Z_D$ cannot have accumulation points in $D_\zeta$: if it had, an argument based on analytic continuations and the  identity theorem would show that $D(z,u_1(z),...,u_c(z))$ is identically zero near $z=0$, which cannot be the case.

Because the path $\gamma$ does not intersect $\Xi[K]\cup Z_D$,  for each $i\in\{1,...,c\}$  the branch $u_i$  can be analytically continued to each point in a simply connected open subset  $W$ s.t. $D_\zeta\setminus (\Xi[K]\cup Z_D)\supseteq W\supseteq\gamma([0,1))$   (monodromy theorem). The same is true for $Q(z,u_1(z),...,u_c(z))$. Both $f(z)$ and $Q(z,u_1(z),...,u_c(z))$ are defined for $z\in W$, and by assumption coincide in   an open connected subset   near $z=0$ which intersects $W$: by the identity theorem, they coincide on $W$, hence on $\gamma$.

By Lemma \ref{lemma:AC}, for $i=1,...,c$,  there is  a (unique) branch of $K$ at $\zeta$, denoted $\tilde u_i :=\AC(K,\zeta,\gamma,u_i)$, analytically connected to $u_i$  along $\gamma$. Moreover, we can choose a common (slit) open disk $D^*$ centered at  $\zeta$ such that all the $\tilde u_i$'s  are  defined and analytic in $D^*$. 
By construction, 
there is a nonempty tail $\tau$ of $\gamma$ s.t.
  $\tau:=\gamma((1-\epsilon, 1))\subseteq D_\zeta\cap D^*$  for some $\epsilon >0$.
  By definition, $u_i(z)=\tilde u_i(z)$  for $z\in \tau$, hence   also $f(z)=Q(z,u_1(z),...,u_c(z))=Q(z,\tilde u_1(z),...,\tilde u_c(z))=S_\zeta(z)$  for $z\in \tau$. This already proves that $S_\zeta$ is analytically connected to $f(z)$. Additionally, being $ D_\zeta\cap D^*$ an open connected set, the identity theorem ensures that $f=S_\zeta$ in $ D_\zeta\cap D^*$. 
  Now concerning part 1,  there are two cases.
\begin{itemize}
\item  $S_\zeta$ is a Taylor series. This defines an analytic function at $z=\zeta$, still called $S_\zeta$ for simplicity, defined   in an open disk $U$ centered at  $\zeta$. We can assume $U$ is small enough that $D_\zeta\cap U \subseteq D_\zeta\cap D^*$.
    We have shown above that  $S_\zeta=f$   in $D_\zeta\cap D^*$, hence in  $D_\zeta\cap U$.  This proves that $f$ can be extended analytically to $\zeta$ (take $f^*=S_\zeta$ in Def. \ref{def:sing}).
\item $S_\zeta$ is not a Taylor series.
    Then $f$ cannot be analytically continued to $\zeta$, in other words $\zeta$ is a singularity of $f$. Indeed, suppose by contradiction that there is $f^*$ analytic in an open disk $U$ centered at $\zeta$ such that $f$ and $f^*$ coincide in $D_\zeta\cap U$. We can assume   $U$ is small enough that $U \subseteq D^*\cup \rho$, where $\rho$ is the slit ray of $D^*$ (not intersecting $D_\zeta$ by assumption), hence $D_\zeta\cap U \subseteq D_\zeta\cap D^*$. We have shown above that  $f=S_\zeta$   in $D_\zeta\cap D^*$, hence in   $D_\zeta\cap U$; but $f=f^*$ in this set, hence by transitivity $S_\zeta=f^*$ in $D_\zeta\cap U$. Since $D_\zeta\cap U\subseteq U\setminus\rho$ and $ U\setminus\rho$  is open connected, by the identity theorem we have $S_\zeta=f^*$ in $ U\setminus\rho$. By the uniqueness of Puiseux series (two Puiseux series in $z-\zeta$ defining the same function in a slit disk of $\zeta$ are the same, term by term), and the fact that $f^*$ has by definition a Taylor series expansion from $\zeta$ valid in $U$, it would follow that $S_\zeta$ is a Taylor series, which is a contradiction.
\end{itemize}
%
Finally, part 2  follows  directly from part 1.
\end{proof_of}


\vsp
\begin{proof_of}{Theorem \ref{th:radius_corr}}
\begin{itemize}
\item \emph{(Part 1: correctness)}.
Let us   assume  that the algorithm does not \textbf{Fail} at step \ref{step:ideal}, otherwise there is nothing to prove. Therefore $\Sigma:=\Xi[K]\cup \Lambda[K,D]$ can be assumed   finite. In the rest of the proof, we let  $f(z):=\Delta^d(z)$.

\begin{enumerate}
\item
We first show that all the dominant singularities of $f$ must belong to $\Sigma$. Indeed,  consider    $\zeta\notin  \Sigma$ such that $|\zeta|=R$.  Suppose by contradiction that $\zeta$ is a singularity. Then Theorem \ref{th:singAng}(1)  implies that $S_\zeta(z)=Q(z,S_1(z),...,S_c(z))$ is \emph{not} a Taylor series in $z$. But each $S_i=\AC(K,\zeta,\gamma,u_i)$ is a Taylor series, because  $\zeta\notin \Xi[K]$ hence the IFT applies (see Section \ref{sec:prelim}). Moreover, $\zeta\notin \Lambda[K,D]$, hence $D(z,S_1(z),...,S_c(z))\neq 0$ in a disk centered at $z=\zeta$, which implies that $S_\zeta(z)=Q(z,S_1(z),...,S_c(z))$  is well defined and analytic in a disk centered at $\zeta$, which contradicts $S_\zeta$ is not a a Taylor series.

\item Since $f(z)$ is a sub-probability distribution function, we must have $R\geq 1$. Moreover, if $R<+\infty$,  $R$ itself must be a dominant singularity (Lem. \ref{lemma:R}(2,b), Pringsheim), hence    $R\in \Sigma_1\cup \cdots \cup \Sigma_t$, for the subsets $\Sigma_i$ defined  in step \ref{step:pring}.

\item If $R<+\infty$  then   $R\in \Sigma_j$, for some $1\leq j \leq t$. Then for each $i<j$, every  $\zeta\in \Sigma_i$ (real or not) is \emph{not} a singularity, i.e. $S_\zeta$ \emph{is} a Taylor series (Th. \ref{th:singAng}(1)), hence steps \ref{step:initS}--\ref{step:return} of the iteration will be skipped, and the next iteration will be   taken. On the contrary, for $i=j$,  we have $R=\zeta\in \Sigma_i$, and the above mentioned steps will be performed, causing the algorithm to return  a set  $S$ as prescribed by the theorem's statement.

\item If $R=+\infty$, then $R\notin\Sigma$: no $\zeta\in \Sigma_i$, for any $i$, is a singularity. Hence, as argued in the previous point, the return at step \ref{step:return} will never be taken. The algorithm will stop at step \ref{step:final} returning $S=\es$ as prescribed by the theorem's statement.
\end{enumerate}
\item \emph{(Part 2: 1-state completeness).}
Let $Q=T/D$ be the rational expression  returned by Algorithm  \ref{alg:ana}. For a  nontrivial 1-state program, $D=1$. Then in step \ref{step:ideal}, $\Lambda[K,D]=\es$, as a consequence $J=\langle 1\rangle\neq \{0\}$, and \textbf{Fail} is not raised.  Note that the $K(z,u)$ returned by Algorithm  \ref{alg:ana} is square-free by definition.
\end{itemize}
\end{proof_of}
\ifmai

\begin{enumerate}
\item \emph{(Correctness)}. First,  we note that the dominant (= of modulus $R$) singularities of $f$, assuming $f$ has any,  must be contained in $\Sigma$ as defined in step 2.
Indeed, by Theorem~\ref{th:singAng}(1), every dominant singularity of \(f\) must be a point where
$
S_Q(z)=Q\bigl(z,S_1(z),\dots,S_c(z)\bigr)
$
fails to be analytic. Since $Q$ is rational,  the set of all such points is  contained in
$
\Sigma=\Xi[K]\cup\Lambda[K,D]
$, where $\Xi[K]$ accounts for candidate singularities of the individual branches $S_i(z)$, and $\Lambda[K,D]$ for candidate zeros of the denominator of $Q$ after substitution.  Therefore no dominant singularity of \(f\) can lie outside \(\Sigma\).

Assume first that \(R<+\infty\). Since \(f\) is a pgf, its coefficients are
nonnegative; hence, by Lemma~\ref{lemma:R}(2b) (Pringsheim theorem), the positive real
number \(R\) is itself a singularity of \(f\). Moreover, as a (sub-)pgf   $f(r)$ converges for every
real $0\leq   r \leq 1$, so we must have $R\geq 1$.  Therefore the subset of \(\Sigma\)
containing the dominant singularities must also contain  a real element \(\ge 1\),
and so is not
discarded at line~2. Let \(\Sigma_j\) be  subset containing elements of modulus $R$. By Lemma~\ref{lemma:R}(1), every dominant singularity has modulus \(R\), hence
all dominant singularities lie in \(\Sigma_j\).
Consider now any \(\zeta\in \Sigma_i\) with \(i<j\). Then \(|\zeta|<R\) and $\zeta$ is not a singularity
by Lemma~\ref{lemma:R}(1). The path used at
line~7 is contained in \(\{z:|z|<|\zeta|\}\subseteq \{z:|z|<R\}\), so
Theorem~\ref{th:singAng} applies with \(f(z)=\Delta^d(z)\). Thus the series \(S_\zeta(z)\)
computed at line~8 is the (truncated) Puiseux series of \(f\) at \(\zeta\). Since
\(\zeta\) is not a singularity, by
Theorem~\ref{th:singAng}(1)  the series \(S_\zeta(z)\) has neither a negative
initial coefficient $n_0$ nor a   ramification index $\kappa>1$, so line~9 does not add \(\zeta\) to \(S\).
Now let \(\zeta\in \Sigma_j\). Since  \(|\zeta|= R\),
Theorem~\ref{th:singAng} applies, and line~8 computes the (truncated) Puiseux series of \(f\) at
\(\zeta\). Then line~9 keeps \(\zeta\) if and only if \(n_0<0\) or \(\kappa>1\), which by
Theorem~\ref{th:singAng}(1) is equivalent to \(\zeta\) being a singularity of \(f\). Since
all singularities of modulus \(R\) are dominant by Lemma~\ref{lemma:R}(1), the set $
S=\{(\zeta_1,S_{\zeta_1}(z)),\dots,(\zeta_k,S_{\zeta_k}(z))\}\neq \emptyset,
$ returned
by the algorithm is exactly the set of dominant singularities of \(f\) paired with
the corresponding truncated Puiseux series \(S_{\zeta_i}(z)\).
%
Moreover, by Lemma~\ref{lemma:R}(1),
$
R=|\zeta_i|$ for every  $i$.

If instead $R=+\infty$, then   no     subset \(\Sigma_i\) contains
a singularity. For $i=1,...,t$ and $\zeta\in \Sigma_i$, we can then repeat the above argument and show that   line~9 does not add  \(\zeta\) to \(S\). Therefore the algorithm runs through all $\Sigma_i$'s and eventually returns $S\es=\es$, correctly so.

\item \emph{(1-state completeness).}
Let $Q=T/D$ be the rational expression  returned by Algorithm  \ref{alg:ana}. For a  nontrivial 1-state program, $D=1$. Then in step 1, $\Lambda[K,D]=\es$, as a consequence $J=\langle 1\rangle\neq \langle 0\rangle$, and \textbf{Fail} is not raised.  Note that $K(z,u)$ is assumed to be square-free.
\end{enumerate}
\end{proof_of}
\fi

\section{Additional details}\label{app:details}

\subsection{Section \ref{sec:alg}: running example}
We illustrate Algorithm \ref{alg:ana} on our running example.
First, we build system \eqref{eq:Meq0}:
{
$$\left\{
\begin{array}{lcll}
M^d_{{1}} \left( z,u \right)&=&1+z \left( (1/2)\,{\frac {M_{
{1}}^d \left( z,u \right) }{u}}+(2/3)\,uM_{{2}}^d \left( z,u \right)
\right) -z{(1/2)\frac {   M_{1,0}^d\left( z \right)}{
u}}\\
M^d_{{2}} \left( z,u \right)&=&z \left( (1/2)\,uM_{{1}}^d \left( z,u
\right) +(1/6)\,{\frac {M_{{2}}^d \left( z,u \right) }{{u}^{2}}} \right)
-  z\left( (1/6)\,{\frac {  M_{{2,0}}^d\left( z \right) }{{u
}^{2}}}+(1/6)\,{\frac {  M_{{2,1}}^d \left( z \right) }{u}}
\right)\,.
\end{array}\right.
$$
}\noindent
In the matrix form { $M^d(z,u)=v^d(u)+zA(u)M^d\replace{(z,u)}-B(z,u)$}, we have { $v^d(u)=[1, 0]^T$},
{ $A(u)=\left[\begin{smallmatrix}1/(2u)\, &\,  (2/3)u\\ (1/2)u\, &\, 1/(6u^2)  \end{smallmatrix}\right]$},
and {$B(z,u)=z\cdot[(1/2)M^d_{1,0}(z)/u  ,\, (1/6)(M^d_{2,0}(z)/\replace{u^2}+ M^d_{2,1}(z)/\replace{u})]^T$}.
We then compute the kernel polynomial $K(z,u)$ (Line \ref{line:l1}), and determine its small branches (Line \ref{line:l2}). In detail,  $\det (I-zA(u))=  1-(1/2) z/u-(1/6) z/u^2+(1/12) z^2/u^3-(1/3) u^2 z^2$, hence $e=3$ and  $K(z,u)=-4 u^5 z^2+12 u^3-6 u^2 z-2 u z+z^2$. The polynomial $K$ has  $c=3$ small branches, $u_1(z),u_2(z),u_3(z)$,  with series (Puiseux) expansions near $z=0$:
{
\[
\begin{aligned}
u_1(z) &= 1/2  z- 1/16  z^6+O(z^7)\\
u_2(z) &=   \sqrt 6/6 z^{\frac 1 2}+O(z^{\frac 7 2})\\
u_3(z) &=- \sqrt 6/6 z^{\frac 1 2}+O(z^{\frac 7 2}).
\end{aligned}
\]}\noindent
Substituting these branches into \eqref{eq:MainEq} yields a linear system. In this case, the system involves three unknowns $\widetilde{M}^d_{j,h}\replace{(z)}=(M^d_{1,0}\replace{(z)},M^d_{2,0}\replace{(z)},M^d_{2,1}\replace{(z)})$ and three independent equations:
{\small
$$\left\{
\begin{array}{ll}
{L(z,u_s(z),\widetilde{M}^d_{j,h}{(z)})  }=&\label{eq:systemrun}
 \left( -1/2\,{\frac {z}{u_s(z)}}+1/12\,{\frac {{z}^{2}}{{u_s(z)}^{3}}}-1/4\,{z}
^{2} \right)    M_{{1,0}}  ( z  ) + \left( -1/9
\,{\frac {{z}^{2}}{u_s(z)}}-1/6\,{\frac {z}{{u_s(z)}^{2}}}+1/12\,{\frac {{z}^{2}
}{{u_s(z)}^{3}}} \right)   M_{{2,0}}  ( z  ) +\\
& \left( -1/9\,{z}^{2}-1/6\,{\frac {z}{u_s(z)}}+1/12\,{\frac {{z}^{2}}{{u_s(z)}^{
2}}} \right)  M_{{2,1}}   ( z  ) +1-1/6\,{
\frac {z}{{u_s(z)}^{2}}}+1/2\,zu \\
{\text{($s=1,2,3$).}}
\end{array}
\right.
$$}\noindent

\noindent
Solving this system yields explicit rational expressions in $z,u_s(z)$ for the $M^d_{j,h}(z)$'s (Line \ref{line:l3}). This in turn
allows us to compute a rational expression $Q=\Num/D$ for $\Delta^d(z)$ from   formula \eqref{eq:run3} in Example \ref{ex:run3} (Line \ref{line:l4}). The resulting expression is too large to be reported here. Expanding it into a Puiseux series, we obtain
{
$$\Delta^0(z)=(1/2)z+(1/6){z}^{2}+(1/36){z}^{4}+{  ({7}/{108})}{z}^{5}+(1/36)\,{z}^{6}
+O \left( {z}^{7} \right).$$}

\subsection{Section \ref{sec:algan}: additional details on Algorithm \ref{alg:radius}.}
We provide additional details on steps \ref{step:AC} and \ref{step:analytic} of Algorithm \ref{alg:radius}.

\paragraph{Step \ref{step:AC}: computing $\AC$}
Computing $\AC(K,\zeta,\gamma,u_i)$ consists in identifying  the   branch   of $K(z,u)$ at $z=\zeta$ that is analytically connected to the (small) branch  $u_i(z)$ at $z=z_0$ near 0, as per Definition \ref{def:conn}.  There are several numerical-symbolic algorithms to do this.  For instance,  one can choose  a \emph{smooth} path
$\gamma:[0,1)\to D_\zeta\setminus  \Xi[K]$ connecting $z_0$ to $\zeta$.
Then one considers an ODE satisfied by $u(\gamma(t))$, obtained by   differentiating $K(\gamma(t),u(\gamma(t)))=0$ (implicit differentiation; below we abbreviate $\frac{\partial}{\partial x}(\cdot)$ as $\partial_x$):
{
$$
\frac{\mathrm{d}}{\mathrm{d}t}\,u(\gamma(t))
= -\frac{\partial_z K(\gamma(t),u(\gamma(t)))}{\partial_u K(\gamma(t),u(\gamma(t)))}\cdot \gamma'(t).
$$}\noindent
Note that the denominator of the fraction on the RHS above cannot be 0 at any $t$, as this would imply that $\partial_u K(z_0, u_0)=0$ and $K(z_0,u_0)=0$ with $z_0=\gamma(t)$ and $u_0=u(\gamma(t))$, that is $z_0\in \Xi[K]$, which is excluded by the choice of $\gamma$. One   numerically integrates  this ODE from  $t=0$  to $t\to 1^{-}$ and tracks the corresponding solutions --- one for each initial value $u_0$ at $t=0$, that is, for each distinct root of $K(z_0,u)=0$.
In this way, one connects each  root  of $K(z_0,u)$ to those of $K(\zeta,u)$ along the distinct branches of $K(z,u)$ at $z=z_0$. Distinct branches can be effectively represented by distinct truncated Puiseux series.

This continuation problem is well studied in numerical algebraic geometry, under the name of  \emph{Homotopy Continuation}.
Algorithms giving guarantees,  and related effective   tools    such as the \emph{Bertini} package, are available: see e.g.  \cite{bates2023numerical,duff2023polynomial} and references therein. This is also the approach followed by {\sc Maples}'s  \textsf{Acontinuation}, an undocumented Maple's procedure used under the hood in the \textsf{algcurves} package.

\paragraph{Step \ref{step:analytic}: deciding if $S_\zeta(z)$ is Taylor}
The problem can be split into two parts: (1) deciding if   $\val(S_\zeta)<0$; (2) deciding if $\ri(S_\zeta)>1$. If either case occurs, $S_\zeta$ is not Taylor, otherwise it is Taylor.

Concerning (1), we first compute the first nonzero term of the denominator series $D(z,S_1(z),...,S_c(z))$ and denote by $d$ its exponent. Note that we are guaranteed that $D(z,S_1(z),...,S_c(z))$ is not identically 0, so truncating the $S_i$'s at a sufficiently high order, we are guaranteed to find such a nonzero term. Next, we expand the numerator series $T(z,S_1(z),...,S_c(z))$   only far enough to determine whether it has a nonzero term of exponent $<d$: if this is the case then $\val(S_\zeta)<0$ (i.e. $S_\zeta$ has a pole at $\zeta$), otherwise $\val(S_\zeta)\geq 0$. Note that $T(z,S_1(z),...,S_c(z))$ hence $S_\zeta$ might well be identically 0.

Concerning (2), we rely on the \emph{local monodromy} property of algebraic functions, see
\cite[Ch.~VII]{Flajo}.  For each convergent Puiseux series $S_i$ at
$\zeta$, let $s_i$ denote the corresponding single-valued
determination on a chosen slit disk about $\zeta$.
Choose \(z_0\)  in this slit disk, and let \(\gamma\) be a simple closed curve based at \(z_0\), winding once around \(\zeta\), whose image avoids the exceptional values of \(K\), and whose interior contains no exceptional value of \(K\) other than possibly \(\zeta\).
%
Starting from \(z_0\), the local determinations \(s_i\) can be analytically continued along \(\gamma\), even if \(\gamma\) leaves their original slit disk: let us denote by $s_i^\gamma$ such continuations.  The continuation of $s_\zeta=Q(z,s_1,\ldots,s_c)$ along $\gamma$ is obtained as  $s^\gamma_\zeta=Q(z,s^\gamma_1,\ldots,s^\gamma_c)$.  All these functions are defined and analytic on a sufficiently small neighbourhood $V$ of $z_0$.
Local monodromy implies that
$$ \ri(S_\zeta)=1 \quad\text{ if and only if}\quad s_\zeta^\gamma=s_\zeta \text{ on }V.$$
In other words,  \(S_\zeta\) has ramification index one exactly when analytic continuation once around \(\zeta\) returns to the same local branch at \(z_0\).
The continuation of  the $s_i$'s, from which that of $s_\zeta$ is built, is obtained by using essentially
the same analytic continuation mechanism introduced earlier for $\AC$.
The continuation machinery underlying
\textsc{Maple}'s \textsf{monodromy} procedure in the
\textsf{algcurves} package can be used to implement this test.


\subsection{Section \ref{sec:ext}: simplification}
Consider the  system  \eqref{eq:Meq}, defining the vector $M^d(z,u)$ of meander gf's. Assuming  this system has a unique solution, i.e. that $I-zA(u)$ is nonsingular, one can build an explicit equation involving just one of the components, say $M^d_1(z,u)$. This equation may involve \emph{fewer} functions $M^d_{j,h}(z)$ than the whole system, hence it may be easier to solve: a smaller number of small branches may be sufficient to solve a  \emph{reduced}  version of the affine system \eqref{eq:system}. Moreover,  one can try to express some of the  $M^d_{j,h}(z)$  for $j\neq 1$  in terms of $M_1(z,u)$, using the following facts: (a) solving partially the system \eqref{eq:Meq},    the $M^d_j(z,u)$'s for $j\neq 1$ can be expressed as $M^d_j(z,u)=\alpha(z,u)\cdot M^d_1(z,u)+\beta(z,u)$, for suitable coefficients
$\alpha$ and $\beta$ (in general, involving some of the $M^d_{j,h}(z)$); (b) for each $j$, the univariate meanders gf's can  all be expressed in terms of the bivariate one
{
\begin{align*}
M^d_{j,h}(z)&= \frac 1{h!} \left[\frac{\partial^h}{\partial u^h}M^d_j(z,u)\right]_{u=0}
\end{align*}}\noindent
(for the proof, see \cite{Flajo}, or simply note that $M^d_{j,h}(z,u)=M^d_{j,0}(z)+ u\cdot M^d_{j,1}(z)+ u^2\cdot M^d_{j,2}(z)+\cdots$).

For the program in Section \ref{sec:ext}, the single equation for $M^d_1(z,u)$ is
{
\begin{align*}
M^d_{{1}} \left( z,u \right) &={\frac {4\,u{z}^{2}-{z}^{2}-6\,z}{8\,{u}^{2}z-12\,u}} M^d_{{1}} \left( z,u \right)+L(z,u,M^d_{1,0}(z), M^d_{2,0}(z), M^d_{2,1}(z))
\end{align*}}\noindent
for a   function $L$ rational in $z,u$ and   linear in the $M^d_{j,h}(z)$'s. After some algebra, this is equivalent to{
\begin{align*}
\frac{1}4\,{\frac {8\,{u}^{2}z-4\,u{z}^{2}+{z}^{2}-12\,u+6\,z}{u \left( 2\,u
z-3 \right) }}
&= L(z,u,M^d_{1,0}(z), M^d_{2,0}(z), M^d_{2,1}(z))\,.
\end{align*}}\noindent
Now we take the numerator of the fraction on the left-hand side as a kernel polynomial:
$$K(z,u):=8\,{u}^{2}z-4\,u{z}^{2}+{z}^{2}-12\,u+6\,z\,.$$
Exploiting the fact that  $M^d_2(z, u) = -\frac 3 2 z u \frac{M^d_1(z, u)}{u z-3}$, after some calculus and some algebra one can express: $M^d_{2,0}(z) = 0$ and $M^d_{2, 1}(z) = \frac 1 2 z M^d_{1,0}(z)$. Taking these two equalities into account, and inserting the small branch $u_1(z)$ in the above equation, we get a one-equation system, where $M^d_{1,0}(z)$ still occurs linearly:
{\begin{align*}
L(z,u_1(z),M^d_{1,0}(z),0, \frac 1 2z M^d_{1,0}(z))&=0\,.
\end{align*}}\noindent
This can be solved linearly for $M^d_{1,0}(z)$, yielding:
{\begin{align*}
M^d_{1,0}(z)&={\frac {24\, \left( u_{{1}} \left( z \right)  \right) ^{d+3}z-36\,
\left( u_{{1}} \left( z \right)  \right) ^{d+2}}{u_{{1}} \left( z
\right) z \left( z+6 \right)  \left( 2\,u_{{1}} \left( z \right) z-3
\right) }}
\end{align*}}\noindent
From this, $M^d_{2,0}(z)$ and $M^d_{2,1}(z)$ can be found, then $M^d_{1}(z,u),M^d_{2}(z,u)$ and finally $M^d_2(z)=M^d_2(z,1)$. The final formula for $\Delta^d$, in terms of $u_1(z)$, is
{
\begin{align*}
\Delta^d(z)&={\frac { \left( 2\,{z}^{2}-14\,z+12 \right)  \left( u_{{1}} \left( z
\right)  \right) ^{d+1}+{z}^{2}}{3\,{z}^{2}-14\,z+12}}
\\[10pt]
&=:Q(z,u_1(z))\,.
\end{align*}}\noindent
The rational expression $Q(z,u_1(z))$ and the small branch of $K(z,u)$ can be subjected to the analysis described in Section  \ref{sec:algan}, which yields  the results illustrated in Section \ref{sec:ext}.


\makeatletter
\AtEndDocument{%
  \clearpage
  \addtocounter{page}{-1}%
  \immediate\write\@auxout{\string\csxdef{lastpage}{\thepage}}}
\makeatother
\end{document}